\documentclass[article]{IEEEtran}

\usepackage{graphicx,colordvi,psfrag}
\usepackage{amsmath,amssymb}
\usepackage{epstopdf}
\usepackage[caption=false]{subfig}
\usepackage{epsfig,cite}
\usepackage{calc,pstricks, pgf, xcolor}
\usepackage{nicefrac}
\usepackage{enumerate}
\usepackage{dsfont}

\newcommand{\by}{\mathbf{y}}
\newcommand{\bYe}{\mathbf{Y}_{\text{eff}}}
\newcommand{\bY}{\mathbf{Y}}
\newcommand{\bx}{\mathbf{x}}
\newcommand{\bX}{\mathbf{X}}

\newcommand{\bZe}{\mathbf{Z}_{\text{eff}}}
\newcommand{\bZ}{\mathbf{Z}}

\newcommand{\bH}{\mathbf{H}}

\newcommand{\ba}{\mathbf{a}}
\newcommand{\bc}{\mathbf{c}}
\newcommand{\bC}{\mathbf{C}}

\newcommand{\bb}{\mathbf{b}}
\newcommand{\bU}{\mathbf{U}}
\newcommand{\bu}{\mathbf{u}}

\newcommand{\CV}{\mathcal{V}}
\newcommand{\bw}{\mathbf{w}}
\newcommand{\bt}{\mathbf{t}}

\newcommand{\bN}{\mathbf{N}}
\newcommand{\bn}{\mathbf{n}}

\newcommand{\bG}{\mathbf{G}}
\newcommand{\bs}{\mathbf{s}}

\newcommand{\bF}{\mathbf{F}}

\newcommand{\ZZ}{\mathbb{Z}}
\newcommand{\RR}{\mathbb{R}}

\newcommand{\Tsnr}{\mathsf{SNR}}

\newcommand{\Ql}{Q_{\Lambda}}
\newcommand{\Mod}{\bmod\Lambda}

\newcommand{\Var}{\mathrm{Var}}
\newcommand{\Cov}{\mathrm{Cov}}

\newcommand{\rank}{\mathop{\mathrm{rank}}}

\newcommand{\cube}{\mathop{\mathrm{CUBE}}}
\newcommand{\Unif}{\mathop{\mathrm{Unif}}}
\newcommand{\Vol}{\mathrm{Vol}}

\newcommand{\m}[1]{\mathcal{#1}}
\newcommand{\Ind}{\mathds{1}}

\newtheorem{theorem}{Theorem}
\newenvironment{proof}[1][Proof]{\noindent\textbf{#1.} }{\ \rule{0.5em}{0.5em}}
\newtheorem{corollary}{Corollary}
\newtheorem{remark}{Remark}

\newtheorem{definition}{Definition}
\newtheorem{lemma}{Lemma}

\newrgbcolor{BoxRed}{1 .6 .6}
\newrgbcolor{BoxBlue}{.8 .8 1}
\newrgbcolor{BoxGreen}{.8 1 .8}
\newrgbcolor{BoxOrange}{1 .65 0}
\newrgbcolor{BoxPurple}{.87 .79 .85}
\newrgbcolor{LineGray}{.6 .6 .6}
\newrgbcolor{LineBlue}{.2 .6 .9}
\newrgbcolor{LineRed}{.9 .3 .3}
\newrgbcolor{WaveBlue}{0 .56 .77}
\newrgbcolor{WaveRed}{.93 .02 .18}
\newrgbcolor{BoxWaveBlue}{.84 .94 .96}
\def\darkgreen{green!55!black}

\begin{document}

\title{A Simple Proof for the Existence of ``Good'' Pairs of Nested Lattices}

\author{Or~Ordentlich and
        Uri~Erez,~\IEEEmembership{Member,~IEEE}
\thanks{The work of O. Ordentlich was supported by the Adams Fellowship Program of the Israel Academy of Sciences and Humanities, a fellowship from The Yitzhak and Chaya Weinstein Research Institute for Signal Processing at Tel Aviv University and the Feder Family Award. The work of U. Erez was supported by by the ISF under Grant 1956/15.}
\thanks{O. Ordentlich and U. Erez are with Tel Aviv University, Tel Aviv, Israel (email: ordent,uri@eng.tau.ac.il).
}
\thanks{The material in this paper was presented in part at the 27th IEEE Convention of Electrical and Electronics Engineers in Israel, Eilat, 2012.}
}
\maketitle

\begin{abstract}
This paper provides a simplified proof for the existence of nested lattice codebooks allowing to achieve the capacity of the additive white Gaussian noise channel, as well as the optimal rate-distortion trade-off for a Gaussian source. The proof is self-contained and relies only on basic probabilistic and geometrical arguments. An ensemble of nested lattices that is different, and more elementary, than the one used in previous proofs is introduced. This ensemble is based on lifting different subcodes of a linear code to the Euclidean space using Construction A. In addition to being simpler, our analysis is less sensitive to the assumption that the additive noise is Gaussian. In particular, for additive ergodic noise channels it is shown that the achievable rates of the nested lattice coding scheme depend on the noise distribution only via its power. Similarly, the nested lattice source coding scheme attains the same rate-distortion trade-off for all ergodic sources with the same second moment.
\end{abstract}

\section{Introduction}
\label{sec:Introduction}

While lattices are the Euclidean space counterpart of linear codes in Hamming space, the two fields historically developed along quite different paths. From the onset of coding theory, linear codes were treated both using algebraic tools as well as via probabilistic methods. While the history of the theory of lattices began much earlier, with the exception of the Minkowski-Hlawka theorem, its development leaned heavily on purely algebraic constructions until quite recently. This has led to a rather convoluted path for arriving at basic proofs for the existence of lattices possessing ``goodness'' properties that are central to communication problems. The goal of this work is to provide a simple proof for the existence of lattices with the minimal ``goodness'' requirements necessary for achieving the capacity of the AWGN channel, as well as the optimal rate-distortion tradeoff for a white Gaussian source.

A major difference between linear codes and lattices is that the former are finite, while the latter are unbounded. As a result, the application of linear codes to communication settings is more straightforward. The application of lattices for communication problems requires intersecting the (infinite) lattice with a finite shaping region, in order to construct a codebook.

For the problem of source coding, it has been recognized early on~\cite{zador} that the significance of the shaping region becomes less crucial as the quantization resolution grows. Indeed, high resolution is the natural operating point in practical systems, and thus neglecting the shaping region and studying the quantization performance of the lattice is sufficient. Namely, the performance of a lattice quantizer at high resolution, is dictated by its normalized second moment. The asymptotic optimality of lattice quantizers in the latter sense, was established in~\cite{zf96}, where the existence of sequences of lattices whose normalized second moment approaches that of ball, was established. Such sequences of lattices are called \emph{good for MSE quantization}. A stronger requirement is that the worst-case squared error distortion attained by a sequence of lattices approaches its average. Sequences of lattices that satisfy this property are called \emph{good for covering} and were shown to exist by Rogers~\cite{rogers59}.

When it comes to channel coding, the equivalent of the high resolution regime is that of high transmission rate. However, communication systems supporting a very large number of information bits per dimension are seldom encountered. As a consequence, it was not until the 1970s that lattice codes were considered for the channel coding problem, starting with the works of Blake~\cite{blake} and de Buda~\cite{deBuda75}, and continuing with~\cite{urbanke,linder,loeliger97}. In these works, the shaping region was naturally taken to be a ball (or a thin spherical shell), which is efficient in terms of power, but results in a codebook with weaker symmetry than the original lattice. Poltyrev~\cite{poltyrev} bypassed this obstacle, by adopting a path analogous to high resolution quantization, and studied the performance of lattices for the unrestricted additive white Gaussian noise (AWGN) channel. In particular, Poltyrev established the existence of sequences of lattices for which the probability of erroneous detection approaches that of AWGN leaving an effective ball whose volume matches the density of the lattice. Such sequences of lattices are called \emph{Poltyrev good}. As a corollary, it follows that there exist sequences of lattices for which the probability of erroneous detection approaches zero as long as the variance of the AWGN is no greater than the squared radius of the effective ball. Such lattices are called \emph{good for channel coding}. We refer the reader to~\cite[Chapter 7]{ramibook} for a more comprehensive definition and treatment of asymptotic goodness properties of lattices.

An alternative approach~\cite{conwaySloane83,Forney89} to using a spherical shaping region, is using a \emph{nested lattice pair} $\Lambda_c\subset\Lambda_f$, where the Voronoi region $\CV_c$ of the lattice $\Lambda_c$ is used for shaping, such that the codebook is $\m{L}=\Lambda_f\cap\CV_c$. This approach has the advantage of retaining the lattice symmetry structure. In particular, it was shown~\cite{ftc00} that there exist sequences of such codebooks that can attain any rate below $\tfrac{1}{2}\log(\Tsnr)$ with \emph{lattice decoding}, i.e., nearest neighbor decoding over the infinite lattice $\Lambda_f$. See also~\cite{ramibook}.

Finally,~\cite{ErezZamirAWGN} introduced a coding scheme using nested lattice pairs in conjunction with MMSE estimation and dithering. This scheme was shown to attain capacity, as well as the Poltyrev error exponent, with lattice decoding. It is worthwhile noting, that the proof hinged on the coarse lattice being good for covering, and the fine lattice being Poltyrev good. A similar MMSE estimation approach for the source coding problem, was shown to achieve the rate-distortion function of a Gaussian source~\cite{nestedLattices}.

The nested lattice coding scheme of~\cite{ErezZamirAWGN}, which is described in detail in Section~\ref{sec:main}, transformed the AWGN to a modulo-additive channel, where the additive noise is a linear \emph{mixture} of AWGN and a dither uniformly distributed over the Voronoi region of the coarse lattice. In order to establish that this scheme achieves the capacity of the AWGN channel, the authors first derived its error exponent, and then obtained the capacity result as a corollary. Their error exponent analysis required showing that the probability density function of the mixture noise is upper bounded by that of AWGN with the same second moment, times some term that becomes insignificant as the dimension increases. This in turn, imposed the requirement that the coarse lattice be good for covering. On the other hand, the interest in error exponents led to the requirement that the fine lattice be Poltyrev good.

Consequently, the proof of the error exponent and capacity results in~\cite{ErezZamirAWGN} required showing the existence of a sequence of nested lattice pairs where the fine lattice is Poltyrev good, and the coarse lattice is Rogers good. To this end, an ensemble of random Construction A lattices, rotated by the generating matrix of a lattice good for covering, was defined and analyzed. The proof therefore relied on the existence of lattices that are good for covering, which made it indirect, complicated, and overly stringent.

In the last decade lattice codes were found to play a new role in network information theory allowing to obtain new achievable rate regions, that are not achievable using the best known random coding schemes, for many problems~\cite{Philosof,compAndForIeee,ncl10,Bresler,CoFTransformFull,ncnc15}. See~\cite{nz14} for a comprehensive survey. The scheme of~\cite{ErezZamirAWGN}, or its variations, plays an important role in many of these new techniques. However, since the capacity region is not known for the majority of problems in network information theory, determining the optimal error exponents is far out of scope. Therefore, it is the capacity result from~\cite{ErezZamirAWGN}, rather than the error exponent one, that is often used in this context.

This paper relaxes the goodness properties required by a nested lattice pair in order to be capacity achieving. Namely, we show that a pair of nested lattices where the fine lattice is good for coding and the coarse lattice good for MSE quantization, suffices to achieve the capacity of the AWGN channel under the scheme from~\cite{ErezZamirAWGN}. In fact we prove a more general result, showing that the scheme from~\cite{ErezZamirAWGN} applied with such nested lattice pairs can reliable achieve any rate smaller than $\tfrac{1}{2}\log(1+\Tsnr)$ over all additive \emph{semi norm-ergodic} noise channels. An analogous result holds for quantization.

The class of semi norm-ergodic processes includes all processes whose empirical variance is almost surely not much greater than the variance. In~\cite{LapidothNNdecoding} Lapidoth showed that i.i.d. Gaussian codebooks with nearest neighbor decoding can achieve any rate smaller than $\tfrac{1}{2}\log(1+\Tsnr)$ over the same class of channels. Our result is therefore the lattice codes analogue of~\cite{LapidothNNdecoding}. Moreover, it immediately implies that many nested lattice based coding schemes for Gaussian networks are in fact robust to the exact statistics of the noise, and merely require it to be semi norm-ergodic.

A key result we obtain, is that a dither uniformly distributed over the Voronoi region of a lattice that is good for MSE quantization is semi norm-ergodic, and moreover, any linear combination of such a dither and semi norm-ergodic noise, is itself semi norm-ergodic.
This enables to relax the goodness for covering requirement of the coarse lattice, to goodness for MSE quantization.

Our analysis also naturally extends to the more practical case, where the coarse lattice is the simple one-dimensional cubic lattice, whereas the fine lattice is a Construction A lattice based on some $p$-ary linear code. We show that for large $p$, the scheme from~\cite{ErezZamirAWGN} can reliably achieve any rate smaller than $\tfrac{1}{2}(1+\Tsnr)-\tfrac{1}{2}\log(2\pi e/12)$ with such a coarse lattice. We further explicitly upper bound the loss incurred by using any finite value of $p$.

Most importantly, we provide a simple, self-contained proof for the existence of nested lattice chains $\Lambda_1^{(n)}\subset\cdots\subset\Lambda_L^{(n)}$, for any finite $L$, where all lattice sequences $\Lambda_1^{(n)},\cdots,\Lambda_L^{(n)}$ are good for MSE quantization and for coding. Although this result is not new, and can be obtained as a simple corollary of~\cite{kp07}, our proof techniques are quite different and considerably simpler.
In particular, we define a novel ensemble of nested lattice chains, based on drawing a random linear $p$-ary code and using Construction A to lift $L$ of its sub-codes to the Euclidean space. This ensemble, which is a direct extension of the enseble of nested linear binary codes proposed by Zamir and Shamai in \cite{zs98}, allows for a direct analysis of the goodness figures of merit of its members. Consequently, our existence proof requires only elementary probabilistic and geometrical arguments.

\section{Preliminaries on Lattice Codes}
\label{sec:Preliminaries}
A lattice $\Lambda$ is a discrete subgroup of $\RR^n$ which is closed under reflection and real addition. Any lattice $\Lambda$ in $\RR^n$ is spanned by some $n\times n$ matrix $\bF$ such that
\begin{align}
\Lambda=\{\mathbf{t}=\bF\mathbf{a}:\mathbf{a}\in\ZZ^n\}.\nonumber
\end{align}
We denote the nearest neighbor quantizer associated with the lattice $\Lambda$ by
\begin{align}
\Ql(\bx)\triangleq\arg\min_{\mathbf{t}\in\Lambda}\|\bx-\mathbf{t}\|.
\label{NNquantizer}
\end{align}
The basic Voronoi region of $\Lambda$, denoted by $\CV$, is the set of all points in $\RR^n$ which are quantized to the zero vector, where ties in~\eqref{NNquantizer} are broken in a systematic manner.
The modulo operation returns the quantization error w.r.t. the lattice,
\begin{align}
\left[\bx\right]\Mod\triangleq\bx-\Ql(\bx),\nonumber
\end{align}
and satisfies the distributive law,
\begin{align}
\big[[\bx]\Mod+\by\big]\Mod=\left[\bx+\by\right]\Mod.\nonumber
\end{align}
Let $V(\Lambda)$ be the volume of a fundamental cell of $\Lambda$, i.e., the volume of $\CV$, and let $\bU$ be a random variable uniformly distributed over $\CV$.
We define the second moment per dimension associated with $\Lambda$ as
\begin{align}
\sigma^2(\Lambda)\triangleq\frac{1}{n}\mathbb{E}\|\bU\|^2=\frac{1}{n}\frac{\int_{\CV}\|\bx\|^2d\bx}{V(\Lambda)}.\nonumber
\end{align}
The normalized second moment (NSM) of a lattice $\Lambda$ is defined by
\begin{align}
G(\Lambda)\triangleq\frac{\sigma^2(\Lambda)}{V(\Lambda)^{\frac{2}{n}}}.\nonumber
\end{align}
Note that this quantity is invariant to scaling of the lattice $\Lambda$.

It is often useful to compare the properties of the Voronoi region $\CV$ with those of a ball.

\vspace{1mm}

\begin{definition}
Let
\begin{align}
\mathcal{B}(\bs,r)\triangleq\left\{\bx\in\RR^n \ : \ \|\bx-\bs\|\leq r\right\},\nonumber
\end{align}
denote the closed $n$-dimensional ball with radius $r$ centered at $\bs$.
We denote the volume of an $n$-dimensional ball with unit radius by $V_n$. In general $V\left(\mathcal{B}(\bs,r)\right)=V_n r^n$. Note that $n V_n^{\frac{2}{n}}$ is monotonically increasing in $n$, and satisfies \mbox{$4\leq n V_n^{\frac{2}{n}}<2\pi e$} for all $n$~\cite{ConwaySloane}, and
\begin{align}
\lim_{n\rightarrow\infty}n V_n^{\frac{2}{n}}={2\pi e}.\label{VnAsymptotic}
\end{align}
\end{definition}

\vspace{1mm}

By the isoperimetric inequality, the ball $\mathcal{B}(\mathbf{0},r)$ has the smallest second moment per dimension out of all (measurable) sets in $\RR^n$ with volume $V_n r^n$, and it is given by
\begin{align}
\sigma^2\left(\mathcal{B}(\mathbf{0},r) \right)&=\frac{1}{n}\frac{1}{V_n r^n}\int_{\bx\in\mathcal{B}(\mathbf{0},r)}\|\bx\|^2 d\bx\nonumber\\
&=\frac{1}{n}\frac{1}{V_n r^n}\int_0^{r}r'^2 d(V_n r'^n)\nonumber\\
&=\frac{1}{n}\frac{1}{V_n r^n} \frac{nV_n r^{n+2}}{n+2}\nonumber\\
&=\frac{r^2}{n+2}.\label{ballsecondmoment}
\end{align}
It follows that $\mathcal{B}(\mathbf{0},r)$ has the smallest possible NSM
\begin{align}
G\left(\mathcal{B}(\mathbf{0},r)\right)=\frac{\sigma^2\left(\mathcal{B}(\mathbf{0},r)\right)}{V^{\tfrac{2}{n}}\left(\mathcal{B}(\mathbf{0},r)\right)}
=\frac{1}{n+2}V_n^{-\frac{2}{n}},
\end{align}
which approaches $1/(2\pi e)$ from above as $n\rightarrow\infty$. Thus, the NSM of any lattice in any dimension satisfies \mbox{$G(\Lambda)\geq 1/(2\pi e)$}.

We define the effective radius $r_{\text{eff}}(\Lambda)$ as the radius of a ball which has the same volume as $\Lambda$, i.e.,
\begin{align}
r^2_{\text{eff}}(\Lambda)\triangleq\frac{V^{\frac{2}{n}}(\Lambda)}{V_n^{\frac{2}{n}}}.\label{reffdef}
\end{align}
Since $\mathcal{B}(\mathbf{0},r_{\text{eff}}(\Lambda))$ has the smallest second moment of all sets in $\RR^n$ with volume $V(\Lambda)$, we have
\begin{align}
\sigma^2(\Lambda)\geq\sigma^2\left(\mathcal{B}(\mathbf{0},r_{\text{eff}}(\Lambda))\right)=\frac{r^2_{\text{eff}}(\Lambda)}{n+2}\label{isoReff}.
\end{align}
Thus,
\begin{align}
r_{\text{eff}}(\Lambda)\leq\sqrt{(n+2)\sigma^2(\Lambda)}.\label{reffbound}
\end{align}
Note that for large $n$ we have
\begin{align}
\frac{r^2_{\text{eff}}(\Lambda)}{n}\approx\frac{V^{\frac{2}{n}}(\Lambda)}{2\pi e}.\nonumber
\end{align}

\vspace{1mm}

\begin{definition}
We say that a sequence in $n$ of random noise vectors $\bZ^{(n)}$ of length $n$ with (finite) effective variance \mbox{$\sigma^2_{\bZ}\triangleq\frac{1}{n}\mathbb{E}\|\bZ^{(n)}\|^2$}, is \emph{semi norm-ergodic} if for any $\epsilon,\delta>0$ and $n$ large enough
\begin{align}
\Pr\left(\bZ^{(n)}\notin\mathcal{B}(\mathbf{0},\sqrt{(1+\delta)n\sigma^2_{\bZ}} \right)\leq\epsilon.\label{normergodicDef}
\end{align}
Note that by the law of large numbers, any i.i.d. noise is semi norm-ergodic. However, even for non i.i.d. noise, the requirement~\eqref{normergodicDef} is not very restrictive.
In the sequel we omit the dimension index, and denote the sequence $\bZ^{(n)}$ simply by $\bZ$.
\end{definition}

\vspace{1mm}

\begin{definition}
The \emph{nearest neighbor decoder} with respect to the lattice $\Lambda$ outputs for every $\by\in\RR^n$ the lattice point $Q_{\Lambda}(\by)$.
\end{definition}

\vspace{2mm}

\begin{definition}
A sequence of lattices $\Lambda^{(n)}$ with growing dimension, satisfying
\begin{align}
\lim_{n \rightarrow \infty} V^{\frac{2}{n}}(\Lambda^{(n)})=\Phi\nonumber
\end{align}
for some $\Phi>0$, is called \emph{good for channel coding in the presence of semi norm-ergodic noise} if for any lattice point $\bt\in\Lambda^{(n)}$, and additive semi norm-ergodic noise $\bZ$ with effective variance\footnote{In~\cite{ErezZamirAWGN} the volume-to-noise ratio (VNR) was defined as $$\mu=\lim_{n\to\infty}V^{\frac{2}{n}}(\Lambda^{(n)})/{2\pi e \sigma^2_{\bZ}}.$$ Thus, the condition $\Phi>{2\pi e \sigma^2_{\bZ}}$ is equivalent to $\text{VNR}>1$.} $\sigma^2_{\bZ}=\frac{1}{n}\mathbb{E}\|\bZ\|^2<\Phi/2\pi e$
\begin{align}
\lim_{n\rightarrow\infty}\Pr\left(Q_{\Lambda^{(n)}}(\bt+\bZ)\neq \bt \right)= 0,\nonumber
\end{align}
That is, the error probability under nearest neighbor decoding in the presence of semi norm-ergodic additive noise $\bZ$ vanishes with $n$ if \mbox{$\lim_{n\to\infty}r^2_{\text{eff}}(\Lambda^{(n)})/n>\sigma^2_{\bZ}$}. For brevity, we simply call such sequences of lattices \emph{good for coding} in the sequel.
\end{definition}

\vspace{1mm}

\begin{definition}
A sequence of lattices $\Lambda^{(n)}$ with growing dimension is called good for mean squared error (MSE) quantization if
\begin{align}
\lim_{n\rightarrow\infty}G\left(\Lambda^{(n)}\right)=\frac{1}{2\pi e}.\nonumber
\end{align}
\end{definition}

\vspace{1mm}

A lattice $\Lambda_c$ is said to be nested in $\Lambda_f$ if $\Lambda_c\subset\Lambda_f$. The lattice $\Lambda_c$ is referred to as the coarse lattice and $\Lambda_f$ as the fine lattice. The \emph{nesting ratio} is defined as $\left(V(\Lambda_c)/V(\Lambda_f)\right)^{1/n}$.

\vspace{1mm}

Next, we define ``good'' pairs of nested lattices. Our definition for the ``goodness'' of nested lattice pairs is different from the one used in~\cite{ErezZamirAWGN}.

\begin{definition}
\label{def:goodpairs}
A sequence of pairs of nested lattices \mbox{$\Lambda^{(n)}_c\subset\Lambda^{(n)}_f$} is called ``good'' if the sequence of lattices $\Lambda^{(n)}_c$ and $\Lambda^{(n)}_f$ are good for both MSE quantization and for coding.
\end{definition}

\begin{remark}
As we shall see in Section~\ref{sec:main}, for the problem of coding over the AWGN channel (or more generally, any additive semi norm-ergodic noise channel), it suffices that $\Lambda_f^{(n)}$ is good for coding and $\Lambda_c^{(n)}$ is good for MSE quantization. In order to achieve the optimal rate-distortion function of a white Gaussian source, the roles are reversed and $\Lambda_f^{(n)}$ should be good for MSE quantization while $\Lambda_c^{(n)}$ is good for coding. A sequence of pairs \mbox{$\Lambda^{(n)}_c\subset\Lambda^{(n)}_f$} that is good according to Definition~\ref{def:goodpairs} is therefore adequate for both problems.
\end{remark}

Our existence proofs are based on Construction A~\cite{ConwaySloane}, as defined next.

\begin{definition}[$p$-ary Construction A]
Let $p$ be a prime number, and let $\bG\in\ZZ_p^{k\times n}$ be a $k\times n$ matrix whose entries are all members of the finite field $\ZZ_p$. The matrix $\bG$ generates a linear $p$-ary code
\begin{align}
\mathcal{C}(\bG)\triangleq\left\{\bx\in\ZZ_p^n \ : \ \bx=[\bw^T\bG]\bmod p \ \ \ \bw\in\ZZ_p^{k}  \right\}\nonumber.
\end{align}
The $p$-ary Construction A lattice induced by the matrix $\bG$ is defined as
\begin{align}
\Lambda(\bG)\triangleq p^{-1}\mathcal{C}(\bG)+\ZZ^n.\nonumber
\end{align}
\label{def:constA}
\end{definition}

\section{Main Results}
\label{sec:main}

Our main result is the following.

\vspace{1mm}

\begin{theorem}
For any finite $L$, $0<\alpha_1<\ldots<\alpha_L<\infty$ there exists a sequence of nested lattice  chains $\Lambda^{(n)}_1\subset\cdots\subset\Lambda^{(n)}_L$ for which
\begin{enumerate}
\item $\Lambda^{(n)}_\ell$ is good for MSE quantization and for coding for all $\ell=1,\ldots,L$;
\item $\lim_{n\to\infty}V^{\tfrac{2}{n}}\left(\Lambda^{(n)}_\ell\right)=2\pi e 2^{-\alpha_\ell}$ for all $\ell=1,\ldots,L$.
\end{enumerate}
\label{thm:main}
\end{theorem}

\vspace{1mm}

For the proof of Theorem~\ref{thm:main}, as given in Section~\ref{sec:ensemble}, we define a novel ensemble of nested lattice chains. This ensemble is defined in Section~\ref{sec:ensemble} and is based on drawing a random linear $p$-ary code and using Construction A to lift $L$ of its sub-codes to the Euclidean space. Theorem~\ref{thm:randomConstA}, stated in Section~\ref{sec:ensemble} and proved in Section~\ref{sec:proof}, shows that with high probability each of these lifted sub-codes possesses the goodness properties. The existence of a sequence of good nested lattice chains then follows from a simple union bound argument.

An immediate corollary of Theorem~\ref{thm:main} is the following.

\vspace{1mm}

\begin{theorem}
For any $P_1>P_2>\cdots>P_L>0$ there exists a sequence of nested lattice chains $\Lambda^{(n)}_1\subset\cdots\subset\Lambda^{(n)}_L$ with the following properties
\begin{enumerate}
\item $\Lambda^{(n)}_\ell$ is good for MSE quantization and for coding for all $\ell=1,\ldots,L$;
\item $\lim_{n\to\infty}\sigma^2\left( \Lambda^{(n)}_\ell\right) = P_\ell$ for all $\ell=1,\ldots,L$;
\item For any $1\leq k<m\leq L$ the sequence of nested lattice codebooks $\m{L}^{(n)}_{km}\triangleq  \Lambda_m^{(n)}\cap\CV^{(n)}_k$ has rate $R^{(n)}_{km}\triangleq\frac{1}{n}\log\left|\m{L}^{(n)}_{km}\right|$ that satisfy\footnote{All logarithms in this paper are to the base $2$, and therefore all rates are expressed in bits per (real) channel use.}
    \begin{align}
    \lim_{n\to\infty}R^{(n)}_{km}=\frac{1}{2}\log\left(\frac{P_k}{P_m}\right).\nonumber
    \end{align}
\end{enumerate}
\label{thm:goodchains}
\end{theorem}

\vspace{1mm}

\begin{proof}
Fix $\alpha_1>0$ and, for any $1<\ell\leq L$, set $\alpha_\ell=\alpha_1+\log\left(\tfrac{P_1}{P_\ell} \right)$. By Theorem~\ref{thm:main} there exists a sequence $\Lambda^{(n)}_1\subset\cdots\subset\Lambda^{(n)}_L$, where all lattices are good for MSE quantization and for coding, and in addition, $\lim_{n\to\infty}V^{\frac{2}{n}}\left(\Lambda^{(n)}_\ell\right)=2\pi e 2^{-\alpha_1}\left(\tfrac{P_\ell}{P_1} \right)$. Scaling all lattices in the sequence by $P_1 2^{\alpha_1}$, we get a sequence of lattices that are good for MSE quantization and coding for which $\lim_{n\to\infty}V^{\frac{2}{n}}\left(\Lambda^{(n)}_\ell\right)=2\pi e P_{\ell}$. Since $\sigma^2(\Lambda)=G(\Lambda)V^{\tfrac{2}{n}}(\Lambda)$, the above implies that $\lim_{n\to\infty}\sigma^2\left(\Lambda^{(n)}_\ell\right)=P_{\ell}$ for all $\ell$. In addition,
\begin{align}
\lim_{n\to\infty} R^{(n)}_{km}&=\frac{1}{2}\log\left(\frac{\lim_{n\to\infty}V^{\tfrac{2}{n}}(\Lambda_k^{(n)})}{\lim_{n\to\infty}V^{\tfrac{2}{n}}(\Lambda_m^{(n)})}\right) \nonumber\\
&=\frac{1}{2}\log\left(\frac{P_k}{P_m}\right),\nonumber
\end{align}
as desired.
\end{proof}

It is important to note that Theorem~\ref{thm:main} and Theorem~\ref{thm:goodchains} can be obtained as a special case of the more general results proved in~\cite{kp07,compAndForIeee}.
These results showed the existence of chains of nested lattices where all latices in the chain are both good for coding and good for covering. Goodness for covering implies goodness for MSE quantization~\cite{zf96,ramibook}, and is therefore a stronger property.
However the existence proofs of such chains are quite complicated, and are not self-contained. In particular, these proofs involve starting with a lattice that is good for covering, whose existence is difficult to establish, and rotating a random Construction A lattice using it.
Our main contribution in this paper is in providing a relatively simple, and self-contained proof for Theorem~\ref{thm:main}, from first principles.

In~\cite{ErezZamirAWGN} it was shown that if $\Lambda$ is good for covering, $\bU$ is an independent random vector uniformly distributed over the Voronoi region of $\Lambda$, and $\bZ$ is AWGN with variance $\sigma^2$, then a linear combination $\alpha\bZ+\beta\bU$ is close in distribution to an AWGN with variance $\alpha^2\sigma^2+\beta^2\sigma^2(\Lambda)$. This property played an important role in the analysis of the AWGN capacity achieving nested lattice scheme of~\cite{ErezZamirAWGN}, namely, the mod-$\Lambda$ scheme.

In order to show that pairs of nested lattices that are good according to Definition~\ref{def:goodpairs} achieve the AWGN capacity under the mod-$\Lambda$ coding scheme introduced in~\cite{ErezZamirAWGN}, we need the following theorem that states that any linear combination of semi norm-ergodic noise and a dither from a lattice that is good for MSE quantization is itself semi norm-ergodic.

\vspace{1mm}

\begin{theorem}
\label{thm:effball}
Let $\bZ=\alpha\bN+\beta\bU$, where $\alpha,\beta\in\RR$, $\bN$ is semi norm-ergodic noise, and $\bU$ is a dither statistically independent of $\bN$, uniformly distributed over the Voronoi region $\CV$ of a lattice $\Lambda$ that is good for MSE quantization. Then, the random vector $\bZ$ is semi norm-ergodic.
\end{theorem}

The proof is given in Section~\ref{sec:mixture}. In~\cite{ErezZamirAWGN} it was shown that a nested lattice codebook $\m{L}=\Lambda_f\cap\CV_c$, based on a pair $\Lambda_c\subset\Lambda_f$ where both lattices are good for covering and Poltyrev good can achieve the capacity (as well as the Poltyrev error exponent) of the AWGN channel under the mod-$\Lambda$ scheme. Theorem~\ref{thm:lapidoth}, stated below, shows that the capacity result continues to hold even if the two lattices $\Lambda_c\subset\Lambda_f$ are only good for MSE quantization and for coding, i.e., good according to Definition~\ref{def:goodpairs}. The existence of such good nested lattice pairs is guaranteed by Theorem~\ref{thm:goodchains}. Theorem~\ref{thm:lapidoth} further extends the main result of~\cite{ErezZamirAWGN} to any additive semi norm-ergodic noise channel.

\begin{theorem}
Consider an additive noise channel $Y=X+N$, where $N$ is a semi norm-ergodic noise process with effective variance $\sigma^2_{\bN}=1$ and the input is subject to the power constraint $\tfrac{1}{n}\mathbb{E}\|\bX^2\|^2<\Tsnr$. For any $R<\tfrac{1}{2}\log(1+\Tsnr)$ there exists a sequence of nested lattice codebooks $\m{L}^{(n)}=\Lambda^{(n)}_f\cap\CV^{(n)}_c$ based on a sequence of good nested lattice pairs $\Lambda_c^{(n)}\subset \Lambda_f^{(n)}$, whose rate approaches $R$ and attains a vanishing error probability under the mod-$\Lambda$ scheme.
\label{thm:lapidoth}
\end{theorem}

\begin{proof}
Fix $0<\epsilon<1$ and let $\Lambda_c^{(n)}\subset \Lambda_f^{(n)}$ be a sequence of good nested lattice pairs with
\begin{align}
&\lim_{n\to\infty} \sigma^2\left(\Lambda_c^{(n)}\right)=\Tsnr,\nonumber\\
&\lim_{n\to\infty} \sigma^2\left(\Lambda_f^{(n)}\right)=(1+\epsilon)\tfrac{\Tsnr}{1+\Tsnr},\nonumber
\end{align}
such that the rate of the sequence of codebooks $\mathcal{L}^{(n)}=\Lambda^{(n)}_f\cap\CV^{(n)}_c$ satisfies
\begin{align}
\lim_{n\to\infty}R^{(n)}=\frac{1}{2}\log\left(\frac{1}{1+\epsilon}(1+\Tsnr)\right).\nonumber
\end{align}
The existence of such a sequence of nested lattice pairs is guaranteed by Theorem~\ref{thm:goodchains}. For brevity, we omit the sequence superscripts in the remainder of the proof, and simply use $\Lambda_c,\CV_c,\Lambda_f,\m{L}$ and $R$.

Next, apply the mod-$\Lambda$ scheme of~\cite{ErezZamirAWGN} with the codebook $\m{L}$. Each of the $2^{nR}$ messages is mapped to a codeword in $\mathcal{L}$. Assume the transmitter wants to send the message $w$ which corresponds to the codeword $\bt\in\mathcal{L}$. It transmits
\begin{align}
\bX=[\bt-\bU]\Mod_c,\nonumber
\end{align}
where $\bU$ is a random dither statistically independent of $\bt$, known to both the transmitter and the receiver, uniformly distributed over $\CV_c$. Due to the Crypto Lemma~\cite[Lemma 1]{ErezZamirAWGN}, $\bX$ is also uniformly distributed over $\CV_c$ and is statistically independent of $\bt$. Thus, the average transmission power is $\tfrac{1}{n}\mathbb{E}\|\bX\|^2=\sigma^2(\Lambda_c)=\Tsnr$.

The receiver scales its observation by a factor $\alpha>0$ to be specified later, adds back the dither $\bU$ and reduces the result modulo the coarse lattice
\begin{align}
\bYe&=\left[\alpha \bY+\bU \right]\Mod_c\nonumber\\
&=\left[\bX+\bU+(\alpha-1)\bX+\alpha\bN\right]\Mod_c\nonumber\\
&=\left[\bt+(\alpha-1)\bX+\alpha\bN\right]\Mod_c\nonumber\\
&=\left[\bt+\bZe\right]\Mod_c,\label{}
\label{Yeff}
\end{align}
where
\begin{align}
\bZe=(\alpha-1)\bX+\alpha\bN
\label{effNoise}
\end{align}
is effective noise, that is statistically independent of $\bt$, with effective variance
\begin{align}
\sigma_{\text{eff}}^2(\alpha)\triangleq\frac{1}{n}\mathbb{E}\|\bZe\|^2=\alpha^2+(1-\alpha)^2\Tsnr.\label{vareff}
\end{align}
Since $\bN$ is semi norm-ergodic, and $\bX$ is uniformly distributed over the Voronoi region of a lattice that is good for MSE quantization, Theorem~\ref{thm:effball} implies that $\bZe$ is semi norm-ergodic with effective variance $\sigma_{\text{eff}}^2(\alpha)$. Setting $\alpha=\Tsnr/(1+\Tsnr)$, such as to minimize $\sigma_{\text{eff}}^2(\alpha)$ results in effective variance $\sigma_{\text{eff}}^2=\Tsnr/(1+\Tsnr)$.

The receiver next computes
\begin{align}
\hat{\bt}&=Q_{\Lambda_f}(\bYe)\nonumber\\
&=Q_{\Lambda_f}(\left[\bt+\bZe\right]\Mod_c)\nonumber\\
&=\left[Q_{\Lambda_f}(\bt+\bZe)\right]\Mod_c,\label{codetNN}
\end{align}
and outputs the message corresponding to $\hat{\bt}$ as its estimate. Since $\Lambda_f$ is good for coding, $\bZe$ is semi norm-ergodic, and
\begin{align}
\lim_{n\to\infty}\frac{V^{\frac{2}{n}}(\Lambda_f)}{2\pi e}=(1+\epsilon)\frac{\Tsnr}{1+\Tsnr}>\sigma^2_{\text{eff}},\nonumber
\end{align}
we have that $\Pr(\hat{\bt}\neq \bt)\to 0$ as $n\to\infty$. Taking $\epsilon\to 0$ completes the proof.
\end{proof}

\begin{remark}
We remark that Theorem~\ref{thm:lapidoth} is analogous to the results of~\cite{LapidothNNdecoding} where it is shown that a Gaussian i.i.d. codebook ensemble with nearest neighbor decoding can attain any rate smaller than $\tfrac{1}{2}\log(1+\Tsnr)$ over an additive semi norm-ergodic noise channel. Our result show that the same rate can be attained using nested lattice codes and the mod-$\Lambda$ scheme.
\end{remark}

\begin{remark}
We have shown that nested lattice pairs that are good according to Definition~\ref{def:goodpairs} suffice to achieve the capacity of the AWGN channel. Similarly, it can be shown that such pairs can attain the optimal rate-distortion tradeoff for the a Gaussian source, as well as the optimal rate-distortion trade-off for the Wyner-Ziv problem, under the scheme from~\cite{zs98,nestedLattices}.
\end{remark}

\begin{remark}
In certain applications, chains of nested lattice codes are used in order to convert a Gaussian multiple access channel (MAC) into an effective modulo-lattice channel whose output is a fine lattice point plus effective noise reduced modulo a coarse lattice. Such a situation arises for example in the compute-and-forward framework~\cite{compAndForIeee}, where a receiver is interested in decoding linear combinations with integer valued coefficients of the codewords transmitted by the different users of the MAC. In such applications, the effective noise is often a linear combination of AWGN and \emph{multiple} statistically independent dithers uniformly distributed over the Voronoi region of the coarse lattice. Corollary~\ref{cor:cofmixturenoise}, stated in Section~\ref{sec:mixture}, shows that such an effective noise is semi norm-ergodic regardless of the number of dithers contributing to it, as long as they are all independent and are induced by lattices that are good for MSE quantization. Consequently, nested lattice chains where all lattices are good for MSE quantization and coding, whose existence is guaranteed by Theorem~\ref{thm:goodchains}, suffice to recover all results from~\cite{Philosof,compAndForIeee,ncl10,Bresler,CoFTransformFull,ncnc15} any many other achievable rate regions based on nested lattice coding schemes. Moreover, the analysis in the proof of Theorem~\ref{thm:lapidoth} assumes that the additive noise is semi norm-ergodic, and not necessarily AWGN. Consequently, using a similar analysis it is possible to extend all the results from~\cite{Philosof,compAndForIeee,ncl10,Bresler,CoFTransformFull,ncnc15} to networks with any semi norm-ergodic additive noise.
\end{remark}

\vspace{1mm}

As evident from the proof of Theorem~\ref{thm:lapidoth}, the main role of the coarse lattice $\Lambda_c$ in the mod-$\Lambda$ scheme is to perform shaping. More specifically, the input to the channel is uniformly distributed on $\CV_c$ and in order to approach capacity, such distribution must approach an AWGN as the dimension grows.

In practice, shaping is often avoided in order to reduce the implementation complexity. However, one can always use a nested lattice codebook where the coarse lattice is the simple one-dimensional cubic (integer) lattice, which is of course, not good for MSE quantization. In fact, many practical communication systems apply a $p$-ary linear code, e.g. turbo or LDPC, mapped to a PAM/QAM constellation. The induced constellation in the Euclidean space can be thought of as a nested lattice codebook $\gamma\Lambda_f\cap\gamma\CV_c$, where $\Lambda_f$ is a Construction A lattice based on the chosen linear code, whereas $\Lambda_c$ is the integer lattice $\ZZ^n$.

The scaling parameter $\gamma$, in this case, is dictated by the power constraint. For example, if the power constraint is $\mathbb{E}(X^2)\leq \Tsnr$ the scaling parameter would be $\gamma=\sqrt{12\Tsnr}$. Since $\gamma\Lambda_c=\gamma\ZZ^n\subseteq\gamma\Lambda_f$, the minimum distance in $\gamma\Lambda_f$ cannot exceed $\sqrt{3\Tsnr}$, and in particular does not grow with the dimension. Thus, $\Pr(Q_{\gamma\Lambda_f}(\bt+\bZe)\neq \bt)$ cannot vanish with the lattice dimension, and consequently $\gamma\Lambda_f$ is not good for coding.\footnote{In the next section we specify the ensemble of nested lattices used for the proof of Theorem~\ref{thm:main}, in which $\gamma$ grows as $\sqrt{n}$ in order to avoid this problem.}

Nevertheless, as evident from~\eqref{codetNN}, an error occurs if and only if the lattice point $Q_{\gamma\Lambda_f}(\bt+\bZe)$ is not in the same coset of $\gamma\Lambda_f/\gamma\Lambda_c$ as $\bt$.

\vspace{1mm}

\begin{definition}
The \emph{coset nearest neighbor decoder} with respect to the nested lattice pair $\Lambda_c\subset\Lambda_f$ outputs for every $\by\in\RR^n$ the lattice point $\left[Q_{\Lambda_f}(\by)\right]\Mod_c$.
\end{definition}

\vspace{1mm}

It follows that the mod-$\Lambda$ scheme succeeds if the coset nearest neighbor decoder finds the correct coset. For the case where the coarse lattice is $\gamma\ZZ^n$ this corresponds to $Q_{\gamma\Lambda_f}(\bt+\bZe)=\bt \bmod\gamma$. See Figure~\ref{fig:cosetdecoder} for an illustration. Note that under coset nearest neighbor decoding, the aforementioned pairs of points in $\gamma\Lambda_f$, whose distance is $\gamma$, do not incur an error. Thus, it may be possible to attain an error probability the vanishes with the dimension using the mod-$\Lambda$ scheme.

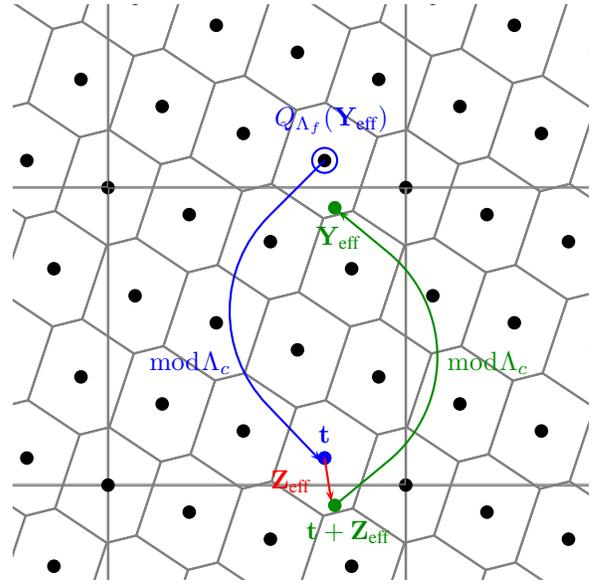
\begin{figure}[h]
\psset{unit=.90mm}
\begin{center}
\begin{pspicture*}(30,30)(115,115)
\rput(0,0){

\rput(0,0){\psframe[linecolor=gray,linewidth=0.5pt](0,0)(44,44)

\rput(0,0){\pscircle[linecolor=black,fillstyle=solid,fillcolor=black](0,0){1}
\psline[linecolor=gray](4.1818,-7.4545)(-0.1818,-8.5455)(-7.8182,-3.4545)(-4.1818,7.4545)(0.1818,8.5455)( 7.8182,3.4545)(4.1818,-7.4545)
}

\rput(8,12){\pscircle[linecolor=black,fillstyle=solid,fillcolor=black](0,0){1}
\psline[linecolor=gray](4.1818,-7.4545)(-0.1818,-8.5455)(-7.8182,-3.4545)(-4.1818,7.4545)(0.1818,8.5455)( 7.8182,3.4545)(4.1818,-7.4545)
}

\rput(16,24){\pscircle[linecolor=black,fillstyle=solid,fillcolor=black](0,0){1}
\psline[linecolor=gray](4.1818,-7.4545)(-0.1818,-8.5455)(-7.8182,-3.4545)(-4.1818,7.4545)(0.1818,8.5455)( 7.8182,3.4545)(4.1818,-7.4545)
}

\rput(24,36){\pscircle[linecolor=black,fillstyle=solid,fillcolor=black](0,0){1}
\psline[linecolor=gray](4.1818,-7.4545)(-0.1818,-8.5455)(-7.8182,-3.4545)(-4.1818,7.4545)(0.1818,8.5455)( 7.8182,3.4545)(4.1818,-7.4545)
}

\rput(32,4){\pscircle[linecolor=black,fillstyle=solid,fillcolor=black](0,0){1}
\psline[linecolor=gray](4.1818,-7.4545)(-0.1818,-8.5455)(-7.8182,-3.4545)(-4.1818,7.4545)(0.1818,8.5455)( 7.8182,3.4545)(4.1818,-7.4545)
}

\rput(40,16){\pscircle[linecolor=black,fillstyle=solid,fillcolor=black](0,0){1}
\psline[linecolor=gray](4.1818,-7.4545)(-0.1818,-8.5455)(-7.8182,-3.4545)(-4.1818,7.4545)(0.1818,8.5455)( 7.8182,3.4545)(4.1818,-7.4545)
}

\rput(4,28){\pscircle[linecolor=black,fillstyle=solid,fillcolor=black](0,0){1}
\psline[linecolor=gray](4.1818,-7.4545)(-0.1818,-8.5455)(-7.8182,-3.4545)(-4.1818,7.4545)(0.1818,8.5455)( 7.8182,3.4545)(4.1818,-7.4545)
}

\rput(12,40){\pscircle[linecolor=black,fillstyle=solid,fillcolor=black](0,0){1}
\psline[linecolor=gray](4.1818,-7.4545)(-0.1818,-8.5455)(-7.8182,-3.4545)(-4.1818,7.4545)(0.1818,8.5455)( 7.8182,3.4545)(4.1818,-7.4545)
}

\rput(20,8){\pscircle[linecolor=black,fillstyle=solid,fillcolor=black](0,0){1}
\psline[linecolor=gray](4.1818,-7.4545)(-0.1818,-8.5455)(-7.8182,-3.4545)(-4.1818,7.4545)(0.1818,8.5455)( 7.8182,3.4545)(4.1818,-7.4545)
}

\rput(28,20){\pscircle[linecolor=black,fillstyle=solid,fillcolor=black](0,0){1}
\psline[linecolor=gray](4.1818,-7.4545)(-0.1818,-8.5455)(-7.8182,-3.4545)(-4.1818,7.4545)(0.1818,8.5455)( 7.8182,3.4545)(4.1818,-7.4545)
}

\rput(36,32){\pscircle[linecolor=black,fillstyle=solid,fillcolor=black](0,0){1}
\psline[linecolor=gray](4.1818,-7.4545)(-0.1818,-8.5455)(-7.8182,-3.4545)(-4.1818,7.4545)(0.1818,8.5455)( 7.8182,3.4545)(4.1818,-7.4545)
}}
\rput(0,44){\psframe[linecolor=gray,linewidth=0.5pt](0,0)(44,44)

\rput(0,0){\pscircle[linecolor=black,fillstyle=solid,fillcolor=black](0,0){1}
\psline[linecolor=gray](4.1818,-7.4545)(-0.1818,-8.5455)(-7.8182,-3.4545)(-4.1818,7.4545)(0.1818,8.5455)( 7.8182,3.4545)(4.1818,-7.4545)
}

\rput(8,12){\pscircle[linecolor=black,fillstyle=solid,fillcolor=black](0,0){1}
\psline[linecolor=gray](4.1818,-7.4545)(-0.1818,-8.5455)(-7.8182,-3.4545)(-4.1818,7.4545)(0.1818,8.5455)( 7.8182,3.4545)(4.1818,-7.4545)
}

\rput(16,24){\pscircle[linecolor=black,fillstyle=solid,fillcolor=black](0,0){1}
\psline[linecolor=gray](4.1818,-7.4545)(-0.1818,-8.5455)(-7.8182,-3.4545)(-4.1818,7.4545)(0.1818,8.5455)( 7.8182,3.4545)(4.1818,-7.4545)
}

\rput(24,36){\pscircle[linecolor=black,fillstyle=solid,fillcolor=black](0,0){1}
\psline[linecolor=gray](4.1818,-7.4545)(-0.1818,-8.5455)(-7.8182,-3.4545)(-4.1818,7.4545)(0.1818,8.5455)( 7.8182,3.4545)(4.1818,-7.4545)
}

\rput(32,4){\pscircle[linecolor=black,fillstyle=solid,fillcolor=black](0,0){1}
\psline[linecolor=gray](4.1818,-7.4545)(-0.1818,-8.5455)(-7.8182,-3.4545)(-4.1818,7.4545)(0.1818,8.5455)( 7.8182,3.4545)(4.1818,-7.4545)
}

\rput(40,16){\pscircle[linecolor=black,fillstyle=solid,fillcolor=black](0,0){1}
\psline[linecolor=gray](4.1818,-7.4545)(-0.1818,-8.5455)(-7.8182,-3.4545)(-4.1818,7.4545)(0.1818,8.5455)( 7.8182,3.4545)(4.1818,-7.4545)
}

\rput(4,28){\pscircle[linecolor=black,fillstyle=solid,fillcolor=black](0,0){1}
\psline[linecolor=gray](4.1818,-7.4545)(-0.1818,-8.5455)(-7.8182,-3.4545)(-4.1818,7.4545)(0.1818,8.5455)( 7.8182,3.4545)(4.1818,-7.4545)
}

\rput(12,40){\pscircle[linecolor=black,fillstyle=solid,fillcolor=black](0,0){1}
\psline[linecolor=gray](4.1818,-7.4545)(-0.1818,-8.5455)(-7.8182,-3.4545)(-4.1818,7.4545)(0.1818,8.5455)( 7.8182,3.4545)(4.1818,-7.4545)
}

\rput(20,8){\pscircle[linecolor=black,fillstyle=solid,fillcolor=black](0,0){1}
\psline[linecolor=gray](4.1818,-7.4545)(-0.1818,-8.5455)(-7.8182,-3.4545)(-4.1818,7.4545)(0.1818,8.5455)( 7.8182,3.4545)(4.1818,-7.4545)
}

\rput(28,20){\pscircle[linecolor=black,fillstyle=solid,fillcolor=black](0,0){1}
\psline[linecolor=gray](4.1818,-7.4545)(-0.1818,-8.5455)(-7.8182,-3.4545)(-4.1818,7.4545)(0.1818,8.5455)( 7.8182,3.4545)(4.1818,-7.4545)
}

\rput(36,32){\pscircle[linecolor=black,fillstyle=solid,fillcolor=black](0,0){1}
\psline[linecolor=gray](4.1818,-7.4545)(-0.1818,-8.5455)(-7.8182,-3.4545)(-4.1818,7.4545)(0.1818,8.5455)( 7.8182,3.4545)(4.1818,-7.4545)
}}
\rput(44,0){\psframe[linecolor=gray,linewidth=0.5pt](0,0)(44,44)

\rput(0,0){\pscircle[linecolor=black,fillstyle=solid,fillcolor=black](0,0){1}
\psline[linecolor=gray](4.1818,-7.4545)(-0.1818,-8.5455)(-7.8182,-3.4545)(-4.1818,7.4545)(0.1818,8.5455)( 7.8182,3.4545)(4.1818,-7.4545)
}

\rput(8,12){\pscircle[linecolor=black,fillstyle=solid,fillcolor=black](0,0){1}
\psline[linecolor=gray](4.1818,-7.4545)(-0.1818,-8.5455)(-7.8182,-3.4545)(-4.1818,7.4545)(0.1818,8.5455)( 7.8182,3.4545)(4.1818,-7.4545)
}

\rput(16,24){\pscircle[linecolor=black,fillstyle=solid,fillcolor=black](0,0){1}
\psline[linecolor=gray](4.1818,-7.4545)(-0.1818,-8.5455)(-7.8182,-3.4545)(-4.1818,7.4545)(0.1818,8.5455)( 7.8182,3.4545)(4.1818,-7.4545)
}

\rput(24,36){\pscircle[linecolor=black,fillstyle=solid,fillcolor=black](0,0){1}
\psline[linecolor=gray](4.1818,-7.4545)(-0.1818,-8.5455)(-7.8182,-3.4545)(-4.1818,7.4545)(0.1818,8.5455)( 7.8182,3.4545)(4.1818,-7.4545)
}

\rput(32,4){\pscircle[linecolor=black,fillstyle=solid,fillcolor=black](0,0){1}
\psline[linecolor=gray](4.1818,-7.4545)(-0.1818,-8.5455)(-7.8182,-3.4545)(-4.1818,7.4545)(0.1818,8.5455)( 7.8182,3.4545)(4.1818,-7.4545)
}

\rput(40,16){\pscircle[linecolor=black,fillstyle=solid,fillcolor=black](0,0){1}
\psline[linecolor=gray](4.1818,-7.4545)(-0.1818,-8.5455)(-7.8182,-3.4545)(-4.1818,7.4545)(0.1818,8.5455)( 7.8182,3.4545)(4.1818,-7.4545)
}

\rput(4,28){\pscircle[linecolor=black,fillstyle=solid,fillcolor=black](0,0){1}
\psline[linecolor=gray](4.1818,-7.4545)(-0.1818,-8.5455)(-7.8182,-3.4545)(-4.1818,7.4545)(0.1818,8.5455)( 7.8182,3.4545)(4.1818,-7.4545)
}

\rput(12,40){\pscircle[linecolor=black,fillstyle=solid,fillcolor=black](0,0){1}
\psline[linecolor=gray](4.1818,-7.4545)(-0.1818,-8.5455)(-7.8182,-3.4545)(-4.1818,7.4545)(0.1818,8.5455)( 7.8182,3.4545)(4.1818,-7.4545)
}

\rput(20,8){\pscircle[linecolor=black,fillstyle=solid,fillcolor=black](0,0){1}
\psline[linecolor=gray](4.1818,-7.4545)(-0.1818,-8.5455)(-7.8182,-3.4545)(-4.1818,7.4545)(0.1818,8.5455)( 7.8182,3.4545)(4.1818,-7.4545)
}

\rput(28,20){\pscircle[linecolor=black,fillstyle=solid,fillcolor=black](0,0){1}
\psline[linecolor=gray](4.1818,-7.4545)(-0.1818,-8.5455)(-7.8182,-3.4545)(-4.1818,7.4545)(0.1818,8.5455)( 7.8182,3.4545)(4.1818,-7.4545)
}

\rput(36,32){\pscircle[linecolor=black,fillstyle=solid,fillcolor=black](0,0){1}
\psline[linecolor=gray](4.1818,-7.4545)(-0.1818,-8.5455)(-7.8182,-3.4545)(-4.1818,7.4545)(0.1818,8.5455)( 7.8182,3.4545)(4.1818,-7.4545)
}}
\rput(44,88){\psframe[linecolor=gray,linewidth=0.5pt](0,0)(44,44)

\rput(0,0){\pscircle[linecolor=black,fillstyle=solid,fillcolor=black](0,0){1}
\psline[linecolor=gray](4.1818,-7.4545)(-0.1818,-8.5455)(-7.8182,-3.4545)(-4.1818,7.4545)(0.1818,8.5455)( 7.8182,3.4545)(4.1818,-7.4545)
}

\rput(8,12){\pscircle[linecolor=black,fillstyle=solid,fillcolor=black](0,0){1}
\psline[linecolor=gray](4.1818,-7.4545)(-0.1818,-8.5455)(-7.8182,-3.4545)(-4.1818,7.4545)(0.1818,8.5455)( 7.8182,3.4545)(4.1818,-7.4545)
}

\rput(16,24){\pscircle[linecolor=black,fillstyle=solid,fillcolor=black](0,0){1}
\psline[linecolor=gray](4.1818,-7.4545)(-0.1818,-8.5455)(-7.8182,-3.4545)(-4.1818,7.4545)(0.1818,8.5455)( 7.8182,3.4545)(4.1818,-7.4545)
}

\rput(24,36){\pscircle[linecolor=black,fillstyle=solid,fillcolor=black](0,0){1}
\psline[linecolor=gray](4.1818,-7.4545)(-0.1818,-8.5455)(-7.8182,-3.4545)(-4.1818,7.4545)(0.1818,8.5455)( 7.8182,3.4545)(4.1818,-7.4545)
}

\rput(32,4){\pscircle[linecolor=black,fillstyle=solid,fillcolor=black](0,0){1}
\psline[linecolor=gray](4.1818,-7.4545)(-0.1818,-8.5455)(-7.8182,-3.4545)(-4.1818,7.4545)(0.1818,8.5455)( 7.8182,3.4545)(4.1818,-7.4545)
}

\rput(40,16){\pscircle[linecolor=black,fillstyle=solid,fillcolor=black](0,0){1}
\psline[linecolor=gray](4.1818,-7.4545)(-0.1818,-8.5455)(-7.8182,-3.4545)(-4.1818,7.4545)(0.1818,8.5455)( 7.8182,3.4545)(4.1818,-7.4545)
}

\rput(4,28){\pscircle[linecolor=black,fillstyle=solid,fillcolor=black](0,0){1}
\psline[linecolor=gray](4.1818,-7.4545)(-0.1818,-8.5455)(-7.8182,-3.4545)(-4.1818,7.4545)(0.1818,8.5455)( 7.8182,3.4545)(4.1818,-7.4545)
}

\rput(12,40){\pscircle[linecolor=black,fillstyle=solid,fillcolor=black](0,0){1}
\psline[linecolor=gray](4.1818,-7.4545)(-0.1818,-8.5455)(-7.8182,-3.4545)(-4.1818,7.4545)(0.1818,8.5455)( 7.8182,3.4545)(4.1818,-7.4545)
}

\rput(20,8){\pscircle[linecolor=black,fillstyle=solid,fillcolor=black](0,0){1}
\psline[linecolor=gray](4.1818,-7.4545)(-0.1818,-8.5455)(-7.8182,-3.4545)(-4.1818,7.4545)(0.1818,8.5455)( 7.8182,3.4545)(4.1818,-7.4545)
}

\rput(28,20){\pscircle[linecolor=black,fillstyle=solid,fillcolor=black](0,0){1}
\psline[linecolor=gray](4.1818,-7.4545)(-0.1818,-8.5455)(-7.8182,-3.4545)(-4.1818,7.4545)(0.1818,8.5455)( 7.8182,3.4545)(4.1818,-7.4545)
}

\rput(36,32){\pscircle[linecolor=black,fillstyle=solid,fillcolor=black](0,0){1}
\psline[linecolor=gray](4.1818,-7.4545)(-0.1818,-8.5455)(-7.8182,-3.4545)(-4.1818,7.4545)(0.1818,8.5455)( 7.8182,3.4545)(4.1818,-7.4545)
}}
\rput(88,0){\psframe[linecolor=gray,linewidth=0.5pt](0,0)(44,44)

\rput(0,0){\pscircle[linecolor=black,fillstyle=solid,fillcolor=black](0,0){1}
\psline[linecolor=gray](4.1818,-7.4545)(-0.1818,-8.5455)(-7.8182,-3.4545)(-4.1818,7.4545)(0.1818,8.5455)( 7.8182,3.4545)(4.1818,-7.4545)
}

\rput(8,12){\pscircle[linecolor=black,fillstyle=solid,fillcolor=black](0,0){1}
\psline[linecolor=gray](4.1818,-7.4545)(-0.1818,-8.5455)(-7.8182,-3.4545)(-4.1818,7.4545)(0.1818,8.5455)( 7.8182,3.4545)(4.1818,-7.4545)
}

\rput(16,24){\pscircle[linecolor=black,fillstyle=solid,fillcolor=black](0,0){1}
\psline[linecolor=gray](4.1818,-7.4545)(-0.1818,-8.5455)(-7.8182,-3.4545)(-4.1818,7.4545)(0.1818,8.5455)( 7.8182,3.4545)(4.1818,-7.4545)
}

\rput(24,36){\pscircle[linecolor=black,fillstyle=solid,fillcolor=black](0,0){1}
\psline[linecolor=gray](4.1818,-7.4545)(-0.1818,-8.5455)(-7.8182,-3.4545)(-4.1818,7.4545)(0.1818,8.5455)( 7.8182,3.4545)(4.1818,-7.4545)
}

\rput(32,4){\pscircle[linecolor=black,fillstyle=solid,fillcolor=black](0,0){1}
\psline[linecolor=gray](4.1818,-7.4545)(-0.1818,-8.5455)(-7.8182,-3.4545)(-4.1818,7.4545)(0.1818,8.5455)( 7.8182,3.4545)(4.1818,-7.4545)
}

\rput(40,16){\pscircle[linecolor=black,fillstyle=solid,fillcolor=black](0,0){1}
\psline[linecolor=gray](4.1818,-7.4545)(-0.1818,-8.5455)(-7.8182,-3.4545)(-4.1818,7.4545)(0.1818,8.5455)( 7.8182,3.4545)(4.1818,-7.4545)
}

\rput(4,28){\pscircle[linecolor=black,fillstyle=solid,fillcolor=black](0,0){1}
\psline[linecolor=gray](4.1818,-7.4545)(-0.1818,-8.5455)(-7.8182,-3.4545)(-4.1818,7.4545)(0.1818,8.5455)( 7.8182,3.4545)(4.1818,-7.4545)
}

\rput(12,40){\pscircle[linecolor=black,fillstyle=solid,fillcolor=black](0,0){1}
\psline[linecolor=gray](4.1818,-7.4545)(-0.1818,-8.5455)(-7.8182,-3.4545)(-4.1818,7.4545)(0.1818,8.5455)( 7.8182,3.4545)(4.1818,-7.4545)
}

\rput(20,8){\pscircle[linecolor=black,fillstyle=solid,fillcolor=black](0,0){1}
\psline[linecolor=gray](4.1818,-7.4545)(-0.1818,-8.5455)(-7.8182,-3.4545)(-4.1818,7.4545)(0.1818,8.5455)( 7.8182,3.4545)(4.1818,-7.4545)
}

\rput(28,20){\pscircle[linecolor=black,fillstyle=solid,fillcolor=black](0,0){1}
\psline[linecolor=gray](4.1818,-7.4545)(-0.1818,-8.5455)(-7.8182,-3.4545)(-4.1818,7.4545)(0.1818,8.5455)( 7.8182,3.4545)(4.1818,-7.4545)
}

\rput(36,32){\pscircle[linecolor=black,fillstyle=solid,fillcolor=black](0,0){1}
\psline[linecolor=gray](4.1818,-7.4545)(-0.1818,-8.5455)(-7.8182,-3.4545)(-4.1818,7.4545)(0.1818,8.5455)( 7.8182,3.4545)(4.1818,-7.4545)
}}
\rput(0,88){\psframe[linecolor=gray,linewidth=0.5pt](0,0)(44,44)

\rput(0,0){\pscircle[linecolor=black,fillstyle=solid,fillcolor=black](0,0){1}
\psline[linecolor=gray](4.1818,-7.4545)(-0.1818,-8.5455)(-7.8182,-3.4545)(-4.1818,7.4545)(0.1818,8.5455)( 7.8182,3.4545)(4.1818,-7.4545)
}

\rput(8,12){\pscircle[linecolor=black,fillstyle=solid,fillcolor=black](0,0){1}
\psline[linecolor=gray](4.1818,-7.4545)(-0.1818,-8.5455)(-7.8182,-3.4545)(-4.1818,7.4545)(0.1818,8.5455)( 7.8182,3.4545)(4.1818,-7.4545)
}

\rput(16,24){\pscircle[linecolor=black,fillstyle=solid,fillcolor=black](0,0){1}
\psline[linecolor=gray](4.1818,-7.4545)(-0.1818,-8.5455)(-7.8182,-3.4545)(-4.1818,7.4545)(0.1818,8.5455)( 7.8182,3.4545)(4.1818,-7.4545)
}

\rput(24,36){\pscircle[linecolor=black,fillstyle=solid,fillcolor=black](0,0){1}
\psline[linecolor=gray](4.1818,-7.4545)(-0.1818,-8.5455)(-7.8182,-3.4545)(-4.1818,7.4545)(0.1818,8.5455)( 7.8182,3.4545)(4.1818,-7.4545)
}

\rput(32,4){\pscircle[linecolor=black,fillstyle=solid,fillcolor=black](0,0){1}
\psline[linecolor=gray](4.1818,-7.4545)(-0.1818,-8.5455)(-7.8182,-3.4545)(-4.1818,7.4545)(0.1818,8.5455)( 7.8182,3.4545)(4.1818,-7.4545)
}

\rput(40,16){\pscircle[linecolor=black,fillstyle=solid,fillcolor=black](0,0){1}
\psline[linecolor=gray](4.1818,-7.4545)(-0.1818,-8.5455)(-7.8182,-3.4545)(-4.1818,7.4545)(0.1818,8.5455)( 7.8182,3.4545)(4.1818,-7.4545)
}

\rput(4,28){\pscircle[linecolor=black,fillstyle=solid,fillcolor=black](0,0){1}
\psline[linecolor=gray](4.1818,-7.4545)(-0.1818,-8.5455)(-7.8182,-3.4545)(-4.1818,7.4545)(0.1818,8.5455)( 7.8182,3.4545)(4.1818,-7.4545)
}

\rput(12,40){\pscircle[linecolor=black,fillstyle=solid,fillcolor=black](0,0){1}
\psline[linecolor=gray](4.1818,-7.4545)(-0.1818,-8.5455)(-7.8182,-3.4545)(-4.1818,7.4545)(0.1818,8.5455)( 7.8182,3.4545)(4.1818,-7.4545)
}

\rput(20,8){\pscircle[linecolor=black,fillstyle=solid,fillcolor=black](0,0){1}
\psline[linecolor=gray](4.1818,-7.4545)(-0.1818,-8.5455)(-7.8182,-3.4545)(-4.1818,7.4545)(0.1818,8.5455)( 7.8182,3.4545)(4.1818,-7.4545)
}

\rput(28,20){\pscircle[linecolor=black,fillstyle=solid,fillcolor=black](0,0){1}
\psline[linecolor=gray](4.1818,-7.4545)(-0.1818,-8.5455)(-7.8182,-3.4545)(-4.1818,7.4545)(0.1818,8.5455)( 7.8182,3.4545)(4.1818,-7.4545)
}

\rput(36,32){\pscircle[linecolor=black,fillstyle=solid,fillcolor=black](0,0){1}
\psline[linecolor=gray](4.1818,-7.4545)(-0.1818,-8.5455)(-7.8182,-3.4545)(-4.1818,7.4545)(0.1818,8.5455)( 7.8182,3.4545)(4.1818,-7.4545)
}}
\rput(88,88){\psframe[linecolor=gray,linewidth=0.5pt](0,0)(44,44)

\rput(0,0){\pscircle[linecolor=black,fillstyle=solid,fillcolor=black](0,0){1}
\psline[linecolor=gray](4.1818,-7.4545)(-0.1818,-8.5455)(-7.8182,-3.4545)(-4.1818,7.4545)(0.1818,8.5455)( 7.8182,3.4545)(4.1818,-7.4545)
}

\rput(8,12){\pscircle[linecolor=black,fillstyle=solid,fillcolor=black](0,0){1}
\psline[linecolor=gray](4.1818,-7.4545)(-0.1818,-8.5455)(-7.8182,-3.4545)(-4.1818,7.4545)(0.1818,8.5455)( 7.8182,3.4545)(4.1818,-7.4545)
}

\rput(16,24){\pscircle[linecolor=black,fillstyle=solid,fillcolor=black](0,0){1}
\psline[linecolor=gray](4.1818,-7.4545)(-0.1818,-8.5455)(-7.8182,-3.4545)(-4.1818,7.4545)(0.1818,8.5455)( 7.8182,3.4545)(4.1818,-7.4545)
}

\rput(24,36){\pscircle[linecolor=black,fillstyle=solid,fillcolor=black](0,0){1}
\psline[linecolor=gray](4.1818,-7.4545)(-0.1818,-8.5455)(-7.8182,-3.4545)(-4.1818,7.4545)(0.1818,8.5455)( 7.8182,3.4545)(4.1818,-7.4545)
}

\rput(32,4){\pscircle[linecolor=black,fillstyle=solid,fillcolor=black](0,0){1}
\psline[linecolor=gray](4.1818,-7.4545)(-0.1818,-8.5455)(-7.8182,-3.4545)(-4.1818,7.4545)(0.1818,8.5455)( 7.8182,3.4545)(4.1818,-7.4545)
}

\rput(40,16){\pscircle[linecolor=black,fillstyle=solid,fillcolor=black](0,0){1}
\psline[linecolor=gray](4.1818,-7.4545)(-0.1818,-8.5455)(-7.8182,-3.4545)(-4.1818,7.4545)(0.1818,8.5455)( 7.8182,3.4545)(4.1818,-7.4545)
}

\rput(4,28){\pscircle[linecolor=black,fillstyle=solid,fillcolor=black](0,0){1}
\psline[linecolor=gray](4.1818,-7.4545)(-0.1818,-8.5455)(-7.8182,-3.4545)(-4.1818,7.4545)(0.1818,8.5455)( 7.8182,3.4545)(4.1818,-7.4545)
}

\rput(12,40){\pscircle[linecolor=black,fillstyle=solid,fillcolor=black](0,0){1}
\psline[linecolor=gray](4.1818,-7.4545)(-0.1818,-8.5455)(-7.8182,-3.4545)(-4.1818,7.4545)(0.1818,8.5455)( 7.8182,3.4545)(4.1818,-7.4545)
}

\rput(20,8){\pscircle[linecolor=black,fillstyle=solid,fillcolor=black](0,0){1}
\psline[linecolor=gray](4.1818,-7.4545)(-0.1818,-8.5455)(-7.8182,-3.4545)(-4.1818,7.4545)(0.1818,8.5455)( 7.8182,3.4545)(4.1818,-7.4545)
}

\rput(28,20){\pscircle[linecolor=black,fillstyle=solid,fillcolor=black](0,0){1}
\psline[linecolor=gray](4.1818,-7.4545)(-0.1818,-8.5455)(-7.8182,-3.4545)(-4.1818,7.4545)(0.1818,8.5455)( 7.8182,3.4545)(4.1818,-7.4545)
}

\rput(36,32){\pscircle[linecolor=black,fillstyle=solid,fillcolor=black](0,0){1}
\psline[linecolor=gray](4.1818,-7.4545)(-0.1818,-8.5455)(-7.8182,-3.4545)(-4.1818,7.4545)(0.1818,8.5455)( 7.8182,3.4545)(4.1818,-7.4545)
}}
\rput(88,44){\psframe[linecolor=gray,linewidth=0.5pt](0,0)(44,44)

\rput(0,0){\pscircle[linecolor=black,fillstyle=solid,fillcolor=black](0,0){1}
\psline[linecolor=gray](4.1818,-7.4545)(-0.1818,-8.5455)(-7.8182,-3.4545)(-4.1818,7.4545)(0.1818,8.5455)( 7.8182,3.4545)(4.1818,-7.4545)
}

\rput(8,12){\pscircle[linecolor=black,fillstyle=solid,fillcolor=black](0,0){1}
\psline[linecolor=gray](4.1818,-7.4545)(-0.1818,-8.5455)(-7.8182,-3.4545)(-4.1818,7.4545)(0.1818,8.5455)( 7.8182,3.4545)(4.1818,-7.4545)
}

\rput(16,24){\pscircle[linecolor=black,fillstyle=solid,fillcolor=black](0,0){1}
\psline[linecolor=gray](4.1818,-7.4545)(-0.1818,-8.5455)(-7.8182,-3.4545)(-4.1818,7.4545)(0.1818,8.5455)( 7.8182,3.4545)(4.1818,-7.4545)
}

\rput(24,36){\pscircle[linecolor=black,fillstyle=solid,fillcolor=black](0,0){1}
\psline[linecolor=gray](4.1818,-7.4545)(-0.1818,-8.5455)(-7.8182,-3.4545)(-4.1818,7.4545)(0.1818,8.5455)( 7.8182,3.4545)(4.1818,-7.4545)
}

\rput(32,4){\pscircle[linecolor=black,fillstyle=solid,fillcolor=black](0,0){1}
\psline[linecolor=gray](4.1818,-7.4545)(-0.1818,-8.5455)(-7.8182,-3.4545)(-4.1818,7.4545)(0.1818,8.5455)( 7.8182,3.4545)(4.1818,-7.4545)
}

\rput(40,16){\pscircle[linecolor=black,fillstyle=solid,fillcolor=black](0,0){1}
\psline[linecolor=gray](4.1818,-7.4545)(-0.1818,-8.5455)(-7.8182,-3.4545)(-4.1818,7.4545)(0.1818,8.5455)( 7.8182,3.4545)(4.1818,-7.4545)
}

\rput(4,28){\pscircle[linecolor=black,fillstyle=solid,fillcolor=black](0,0){1}
\psline[linecolor=gray](4.1818,-7.4545)(-0.1818,-8.5455)(-7.8182,-3.4545)(-4.1818,7.4545)(0.1818,8.5455)( 7.8182,3.4545)(4.1818,-7.4545)
}

\rput(12,40){\pscircle[linecolor=black,fillstyle=solid,fillcolor=black](0,0){1}
\psline[linecolor=gray](4.1818,-7.4545)(-0.1818,-8.5455)(-7.8182,-3.4545)(-4.1818,7.4545)(0.1818,8.5455)( 7.8182,3.4545)(4.1818,-7.4545)
}

\rput(20,8){\pscircle[linecolor=black,fillstyle=solid,fillcolor=black](0,0){1}
\psline[linecolor=gray](4.1818,-7.4545)(-0.1818,-8.5455)(-7.8182,-3.4545)(-4.1818,7.4545)(0.1818,8.5455)( 7.8182,3.4545)(4.1818,-7.4545)
}

\rput(28,20){\pscircle[linecolor=black,fillstyle=solid,fillcolor=black](0,0){1}
\psline[linecolor=gray](4.1818,-7.4545)(-0.1818,-8.5455)(-7.8182,-3.4545)(-4.1818,7.4545)(0.1818,8.5455)( 7.8182,3.4545)(4.1818,-7.4545)
}

\rput(36,32){\pscircle[linecolor=black,fillstyle=solid,fillcolor=black](0,0){1}
\psline[linecolor=gray](4.1818,-7.4545)(-0.1818,-8.5455)(-7.8182,-3.4545)(-4.1818,7.4545)(0.1818,8.5455)( 7.8182,3.4545)(4.1818,-7.4545)
}}

\rput(44,44){\psframe[linecolor=gray,linewidth=0.5pt](0,0)(44,44)

\rput(0,0){\pscircle[linecolor=black,fillstyle=solid,fillcolor=black](0,0){1}
\psline[linecolor=gray](4.1818,-7.4545)(-0.1818,-8.5455)(-7.8182,-3.4545)(-4.1818,7.4545)(0.1818,8.5455)( 7.8182,3.4545)(4.1818,-7.4545)
}

\rput(8,12){\pscircle[linecolor=black,fillstyle=solid,fillcolor=black](0,0){1}
\psline[linecolor=gray](4.1818,-7.4545)(-0.1818,-8.5455)(-7.8182,-3.4545)(-4.1818,7.4545)(0.1818,8.5455)( 7.8182,3.4545)(4.1818,-7.4545)
}

\rput(16,24){\pscircle[linecolor=black,fillstyle=solid,fillcolor=black](0,0){1}
\psline[linecolor=gray](4.1818,-7.4545)(-0.1818,-8.5455)(-7.8182,-3.4545)(-4.1818,7.4545)(0.1818,8.5455)( 7.8182,3.4545)(4.1818,-7.4545)
}

\rput(24,36){\pscircle[linecolor=black,fillstyle=solid,fillcolor=black](0,0){1}
\psline[linecolor=gray](4.1818,-7.4545)(-0.1818,-8.5455)(-7.8182,-3.4545)(-4.1818,7.4545)(0.1818,8.5455)( 7.8182,3.4545)(4.1818,-7.4545)
}

\rput(32,4){\pscircle[linecolor=black,fillstyle=solid,fillcolor=black](0,0){1}
\psline[linecolor=gray](4.1818,-7.4545)(-0.1818,-8.5455)(-7.8182,-3.4545)(-4.1818,7.4545)(0.1818,8.5455)( 7.8182,3.4545)(4.1818,-7.4545)
}

\rput(40,16){\pscircle[linecolor=black,fillstyle=solid,fillcolor=black](0,0){1}
\psline[linecolor=gray](4.1818,-7.4545)(-0.1818,-8.5455)(-7.8182,-3.4545)(-4.1818,7.4545)(0.1818,8.5455)( 7.8182,3.4545)(4.1818,-7.4545)
}

\rput(4,28){\pscircle[linecolor=black,fillstyle=solid,fillcolor=black](0,0){1}
\psline[linecolor=gray](4.1818,-7.4545)(-0.1818,-8.5455)(-7.8182,-3.4545)(-4.1818,7.4545)(0.1818,8.5455)( 7.8182,3.4545)(4.1818,-7.4545)
}

\rput(12,40){\pscircle[linecolor=black,fillstyle=solid,fillcolor=black](0,0){1}
\psline[linecolor=gray](4.1818,-7.4545)(-0.1818,-8.5455)(-7.8182,-3.4545)(-4.1818,7.4545)(0.1818,8.5455)( 7.8182,3.4545)(4.1818,-7.4545)
}

\rput(20,8){\pscircle[linecolor=black,fillstyle=solid,fillcolor=black](0,0){1}
\psline[linecolor=gray](4.1818,-7.4545)(-0.1818,-8.5455)(-7.8182,-3.4545)(-4.1818,7.4545)(0.1818,8.5455)( 7.8182,3.4545)(4.1818,-7.4545)
}

\rput(28,20){\pscircle[linecolor=black,fillstyle=solid,fillcolor=black](0,0){1}
\psline[linecolor=gray](4.1818,-7.4545)(-0.1818,-8.5455)(-7.8182,-3.4545)(-4.1818,7.4545)(0.1818,8.5455)( 7.8182,3.4545)(4.1818,-7.4545)
}

\rput(36,32){\pscircle[linecolor=black,fillstyle=solid,fillcolor=black](0,0){1}
\psline[linecolor=gray](4.1818,-7.4545)(-0.1818,-8.5455)(-7.8182,-3.4545)(-4.1818,7.4545)(0.1818,8.5455)( 7.8182,3.4545)(4.1818,-7.4545)
}
\pscircle[linecolor=blue,fillstyle=solid,fillcolor=blue](32,4){1}
\rput(32,7.5){$\textcolor{blue}{\mathbf{t}}$}
\psline[linecolor=red,linewidth=0.8pt,arrows=->](32,4)(33,-2.5)
\rput(27,0.5){$\textcolor{red}{\mathbf{\bZe}}$}
\pscircle[linecolor=\darkgreen,fillstyle=solid,fillcolor=\darkgreen](33.5,-3){1}
\rput(35.5,-7){$\textcolor{\darkgreen}{\mathbf{t}+\bZe}$}
\pscircle[linecolor=\darkgreen,fillstyle=solid,fillcolor=\darkgreen](33.5,41){1}
\rput(34,36.5){$\textcolor{\darkgreen}{\bYe}$}
\psline[linecolor=\darkgreen,linearc=20,arrows=->,linewidth=0.8pt](33.5,-3)(60,19)(34,40.5)
\rput(56,18){$\textcolor{\darkgreen}{\Mod_c}$}

\pscircle[linecolor=blue](32,48){2}
\rput(33,54){$\textcolor{blue}{Q_{\Lambda_f}(\bYe)}$}

\psline[linecolor=blue,linearc=20,arrows=->,linewidth=0.8pt](32,48)(10,26)(31.5,3.5)
\rput(12,18){$\textcolor{blue}{\Mod_c}$}

}

}
\end{pspicture*}
\end{center}
\caption{An illustration of the coset nearest neighbor decoding process. The lattice point $\bt$ was transmitted. The output of the induced channel when the mod-$\Lambda$ transmission scheme is applied is \mbox{$\bYe=\left[\bt+\bZe\right]\bmod \gamma\ZZ^n$}. The decoder quantizes $\bYe$ to the nearest lattice point in $\Lambda$ and reduces the quantized output modulo $\gamma\ZZ^n$.} \label{fig:cosetdecoder}
\end{figure}

The next theorem, proved in Section~\ref{sec:noshaping}, shows that this is indeed the case. More specifically, it shows that the mod-$\Lambda$ scheme with nested lattice codes where $\Lambda_c=\sqrt{12\Tsnr}\ZZ^n$ can attain any rate smaller than $\tfrac{1}{2}\log(1+\Tsnr)-\tfrac{1}{2}\log(2\pi e/12)$ with a vanishing error probability, if $p$ is large. For finite $p$, an explicit upper bound on the additional loss is also specified.  Let $\cube\triangleq[-1/2,1/2)^n$ denote the unit cube centered at the origin.

\begin{theorem}
Consider an additive noise channel $Y=X+N$, where $N$ is an i.i.d. noise process with unit variance and the input is subject to the power constraint $\tfrac{1}{n}\mathbb{E}\|\bX^2\|^2\leq\Tsnr$. Let
\begin{align}
\Gamma(p,\Tsnr)\triangleq\log\left(1+\sqrt{\frac{3\Tsnr}{p^2}} \right).\nonumber
\end{align}
For any
\begin{align}
R<\frac{1}{2}\log(1+\Tsnr)-\frac{1}{2}\log\left(\frac{2\pi e}{12}\right)-\Gamma(p,\Tsnr),\nonumber
\end{align}
there exists a sequence of nested lattice codebooks $\m{L}^{(n)}=\Lambda^{(n)}_f\cap\sqrt{12\Tsnr}\cdot\cube$ with rate $R$, where $\Lambda^{(n)}_f$ is a sequence of scaled $p$-ary Construction A lattices, that attains a vanishing error probability under the mod-$\Lambda$ scheme.
\label{thm:cubic}
\end{theorem}

Note that $\Gamma(p,\Tsnr)\to 0$ as $p\to\infty$, and the gap to capacity is in this case just the standard shaping loss of $\tfrac{1}{2}\log(2\pi e/12)$. We further note that for any $\epsilon>0$ the choice
\begin{align}
\log{p}>\frac{1}{2}\log(\Tsnr)+\frac{1}{2}\log(3)-\frac{1}{2}\log\left(2^{\epsilon}-1\right),\label{pvalue}
\end{align}
guarantees that $\Gamma(p,\Tsnr)<\epsilon$.

\section{An Ensemble for Nested Lattice Chains}
\label{sec:ensemble}
Previous proofs for the existence of capacity achieving pairs of nested lattices used random Construction A, introduced by Loeliger~\cite{loeliger97}, for creating a fine lattice, and then rotated it using a lattice that is good for covering. Here, we take a different approach that is a direct extension of the original approach of~\cite{zs98} to creating nested binary linear codes. We use random Construction A to simultaneously create both the fine and the coarse lattice. Namely, we randomly draw a linear code and lift it to the Euclidean space in order to obtain the fine lattice. The coarse lattice is obtained by lifting a subcode from the same linear code to the Euclidean space. 

Let $\bG\in\ZZ_p^{k\times n}$. For any natural number $m\leq k$ we denote by $\bG_{m}$ the $m\times n$ matrix obtained by taking only the first $m$ rows of $\bG$. The linear code $\m{C}\left(\bG_{m}\right)$ and the lattice $\Lambda\left(\bG_{m}\right)$ are defined as in Definition~\ref{def:constA}.

Clearly, for any $\bG\in\ZZ_p^{k\times n}$, and $k_1< k$ we have that $\Lambda\left(\bG_{k_1}\right)\subset\Lambda\left(\bG_{k}\right)$. Thus, we can define an ensemble of nested lattice pairs by fixing $k_1,k,n,p$ and drawing the entries of the matrix $\bG$ according to the i.i.d. uniform distribution on $\ZZ_p$.

\begin{remark}
We have chosen to specify our ensemble in terms of the linear codes' generating matrices
\begin{align}
\bG=\left[
        \begin{array}{c}
          \bG_{k_1} \\
          ---\\
          \bG' \\
        \end{array}
      \right].\nonumber
\end{align}
We could have equally defined the ensemble using the linear codes' parity check matrices
\begin{align}
\bH_{n-k_1}=\left[
        \begin{array}{c}
          \bH_{n-k} \\
          ---\\
          \bH' \\
        \end{array}
      \right],\nonumber
\end{align}
as done in~\cite{zs98,nestedLattices} for ensembles of nested binary linear codes.
\end{remark}

More generally, for any choice of $L$ natural numbers $k_1<k_2<\cdots<k_L<n$ we can define a similar ensemble for a chain of $L$ nested lattices
\begin{align}
\Lambda\left(\bG_{k_1}\right)\subset\Lambda\left(\bG_{k_2}\right)\subset\cdots\subset\Lambda\left(\bG_{k_L}\right).\nonumber
\end{align}
We now formally define the ensemble of nested lattices we will use in our existence proof

\begin{definition}[Ensemble of nested lattice chains]
Let $n$ be a natural number and $0<\alpha_1<\ldots<\alpha_L<\log{n}$. An $(n,\alpha_1,\ldots,\alpha_L)$ ensemble for a chain of $L$ nested lattices is defined as follows. Let $\gamma=2\sqrt{n}$, and $p=\xi n^{\frac{3}{2}}$, where $\xi$ is chosen as the smallest number in the interval \mbox{$[1,2)$} such that $p$ is prime. Let
\begin{align}
k_\ell\triangleq\frac{n}{2\log p}\left( \log\left(\frac{4}{V_n^{\frac{2}{n}}} \right)+\alpha_\ell\right), \ \ \ \ell=1,\ldots,L.\nonumber
\end{align}
Draw a matrix $\bG\in\ZZ_p^{k_L\times n}$ whose entries are i.i.d. uniformly distributed over $\ZZ_p$, and construct the chain $\Lambda_1\subset\cdots\subset\Lambda_L$ by setting
\begin{align}
\Lambda_\ell=\gamma\Lambda\left(\bG_{k_\ell}\right), \ \ \ \ell=1,\ldots,L.\nonumber
\end{align}
\end{definition}

Theorem~\ref{thm:main} will follow as a straightforward corollary of the following result.
\begin{theorem}
Let $n$ be a large natural number, $\gamma=2\sqrt{n}$, and $p=\xi n^{\frac{3}{2}}$, where $\xi$ is chosen as the smallest number in the interval \mbox{$[1,2)$} such that $p$ is prime. Further, let $0<\alpha<\log{n}$ and set
\begin{align}
k\triangleq\frac{n}{2\log p}\left( \log\left(\frac{4}{V_n^{\frac{2}{n}}} \right)+\alpha\right).\nonumber
\end{align}
Let $\bG\in\ZZ_p^{k\times n}$ be a random matrix whose entries are i.i.d. uniformly distributed over $\ZZ_p$. Then for any $\epsilon, \delta>0$, there is an integer $N(\epsilon,\delta)$ such that for any $n>N(\epsilon,\delta)$
\begin{enumerate}
\item $\Pr\left(\rank(\bG)<k\right)<\epsilon$;\label{fullrank}
\item $\Pr\left(\sigma^2\left(\gamma\Lambda(\bG)\right)>(1+\delta)2^{-\alpha}\right)<\epsilon$;\label{goodNSM}
\item For any additive semi norm-ergodic noise $\bZ$ with effective variance $\sigma_\bZ^2=\frac{1}{n}\mathbb{E}\|\bZ\|^2\leq (1-\delta)2^{-\alpha}$ and any $\bt\in\gamma\Lambda(\bG)$, the following holds
    \begin{align}
    \Pr\left(\Pr\left(Q_{\gamma\Lambda(\bG)}(\bt+\bZ)\neq \bt \mid \bG\right)>\delta \right)<\epsilon.\nonumber
    \end{align}\label{goodForCoding}
\end{enumerate}
\label{thm:randomConstA}
\end{theorem}

The proof of Theorem~\ref{thm:randomConstA} is given in Section~\ref{sec:proof}. We now prove Theorem~\ref{thm:main}.

\begin{proof}[Proof of Theorem~\ref{thm:main}]
Let $\Lambda_1\subset\cdots\subset\Lambda_L$ be a random lattice chain drawn from the $(n,\alpha_1,\ldots,\alpha_L)$ ensemble and let $\bG_{k_1},\ldots,\bG_{k_L}$ be the corresponding linear codes generating matrices. Set $\epsilon,\delta>0$ and for all $\ell=1,\ldots,L$ define the following error events
\begin{enumerate}
\item $E_{1\ell}$ is the event that $\rank\left(\bG_{k_\ell}\right)<k_{\ell}$;
\item $E_{2\ell}$ is the event that $\sigma^2\left(\Lambda_\ell\right)>(1+\delta)2^{-\alpha_\ell}$
\item $E_{3\ell}$ is the event that $\Pr\left(Q_{\Lambda_\ell}(\bt+\bZ)\neq \bt\right)>\delta$ for some $\bt\in\Lambda_\ell$ and some additive semi norm-ergodic noise $\bZ$ with effective variance $\sigma_\bZ^2=\frac{1}{n}\mathbb{E}\|\bZ\|^2\leq (1-\delta)2^{-\alpha}$
\end{enumerate}
Further, let
\begin{align}
E\triangleq\bigcup_{i=1}^3\bigcup_{\ell=1}^L E_{i\ell}.\nonumber
\end{align}
By the union bound we have that
\begin{align}
\Pr(E)&\leq\sum_{i=1}^3 \Pr\left(\bigcup_{\ell=1}^L E_{i\ell}\right)\nonumber\\
&=\Pr(E_{1L})+\Pr\left(\bigcup_{\ell=1}^L E_{2\ell}\right)+\Pr\left(\bigcup_{\ell=1}^L E_{3\ell}\right)\label{eq:rankevent}\\
&\leq\Pr(E_{1L})+\sum_{\ell=1}^L\Pr(E_{2\ell})+\sum_{\ell=1}^L\Pr(E_{3\ell})\nonumber,
\end{align}
where~\eqref{eq:rankevent} follows from the fact that if $\bG_{k_L}$ has full row rank over $\ZZ_p$, then so are all the matrices obtained by removing rows from it. Further, since $\bG_{k_\ell}$ satisfies the conditions of Theorem~\ref{thm:randomConstA} for all $\ell=1,\ldots,L$, then for $n$ large enough $\Pr(E_{1L})<\epsilon$, $\Pr(E_{2\ell})<\epsilon$ and $\Pr(E_{3\ell})<\epsilon$. Thus, $\Pr(E)<(2L+1)\epsilon$, and consequently $\Pr(\overline{E})>1-(2L+1)\epsilon$, where $\overline{E}$ is the event that $E$ did not occur. Since this holds for any $\epsilon>0$, we have that for $n$ large enough, the event $E$ does not occur for almost all members in the ensemble.

We now show that any member in the ensemble for which $E$ does not occur, has lattices $\Lambda_1\subset\cdots\subset\Lambda_L$ whose volumes are close to $2\pi e 2^{-\alpha_\ell}$, whose normalized second moments are close to $1/2\pi e$, and whose error probabilities are small as long as the volume-to-noise ratio is greater than $1$.

In particular, if $E$ does not occur, all matrices $\bG_{k_1},\ldots,\bG_{k_L}$ have full row rank over $\ZZ_p$. In this case, we have that $V(\Lambda_\ell)=\gamma^n p^{-k_{\ell}}$ and therefore
\begin{align}
V^{\tfrac{2}{n}}(\Lambda_\ell)&=\gamma^2 p^{-\frac{2k_{\ell}}{n}}\nonumber\\
&=4n\frac{V_n^{\tfrac{2}{n}}}{4}2^{-\alpha_\ell}.\label{normalizedVolume}
\end{align}
Since $\lim_{n\to\infty} nV_n^{\tfrac{2}{n}}=2\pi e$, we have that $\lim_{n\to\infty}V^{\tfrac{2}{n}}(\Lambda_\ell)=2\pi e 2^{-\alpha_\ell}$, as desired. In particular, for $n$ large enough
\begin{align}
(1-\delta/2)2\pi e 2^{-\alpha_\ell}<V^{\tfrac{2}{n}}(\Lambda_\ell)<2\pi e 2^{-\alpha_\ell}.\label{VolBounds}
\end{align}
Now, by Theorem~\ref{thm:randomConstA}
\begin{align}
G(\Lambda_\ell)&=\frac{\sigma^2(\Lambda_\ell)}{V^{\tfrac{2}{n}}(\Lambda_\ell)}\nonumber\\
&\leq \frac{(1+\delta)2^{-\alpha_\ell}}{(1-\delta/2)2\pi e 2^{-\alpha_\ell}}\nonumber\\
&=(1+\delta')\frac{1}{2\pi e},\nonumber
\end{align}
where $\delta'=(1+\delta)/(1-\delta/2)$ can be made as small as desired by increasing $n$. Thus, the sequence $\Lambda^{(n)}_\ell$ is good for MSE quantization.

In addition (again by Theorem~\ref{thm:randomConstA}, part \ref{goodForCoding}, and~\eqref{VolBounds}), we have that for any semi norm-ergodic noise $\bZ_\ell$ with effective variance $\sigma^2_{\bZ_\ell}\leq (1-\delta)2^{-\alpha_\ell}$, the probability of error in nearest neighbor decoding is smaller than $\delta$. Thus, for $n$ large enough we have that as long as the ratio $V^{\tfrac{2}{n}}(\Lambda_\ell)/(2\pi e \sigma^2_{\bZ_\ell})$ is greater than $(1-\delta/2)/(1-\delta)$, the error probability in decoding a point from $\Lambda_\ell$ in the presence of additive noise $\bZ_{\ell}$ is smaller than $\delta$. Thus, the sequence $\Lambda^{(n)}_\ell$ is good for coding.
\end{proof}

\section{Proof of Theorem~\ref{thm:randomConstA}}
\label{sec:proof}

Before going into the proof we need to introduce some more notation. Denote the operation of reducing each component of $\bx\in\RR^n$ modulo $\gamma$ by $\bx^*\triangleq[\bx]\bmod \gamma\ZZ^n$. If $\mathcal{S}$ is a set of points in $\RR^n$, $\mathcal{S}^*$ is the set obtained by reducing all points in $\mathcal{S}$ modulo $\gamma \ZZ^n$. If $\mathcal{S}$ and $\mathcal{T}$ are sets, $\mathcal{S}+\mathcal{T}$ is their Minkowski sum.
In the sequel, we use the following lemma, which follows from simple geometric arguments and is illustrated in Figure~\ref{fig:spheregrid}.

\begin{figure}[h]
\psset{unit=.8mm}
\begin{center}
\begin{pspicture}(5,5)(95,95)

\rput(25,35){\input{cube}}
\rput(25,45){\input{cube}}
\rput(25,55){\input{cube}}
\rput(25,65){\input{cube}}

\rput(35,25){\input{cube}}
\rput(35,35){\input{cube}}
\rput(35,45){\input{cube}}
\rput(35,55){\input{cube}}
\rput(35,65){\input{cube}}
\rput(35,75){\input{cube}}

\rput(45,25){\input{cube}}
\rput(45,35){\input{cube}}
\rput(45,45){\input{cube}}
\rput(45,55){\input{cube}}
\rput(45,65){\input{cube}}
\rput(45,75){\input{cube}}

\rput(55,25){\input{cube}}
\rput(55,35){\input{cube}}
\rput(55,45){\input{cube}}
\rput(55,55){\input{cube}}
\rput(55,65){\input{cube}}
\rput(55,75){\input{cube}}

\rput(65,25){\input{cube}}
\rput(65,35){\input{cube}}
\rput(65,45){\input{cube}}
\rput(65,55){\input{cube}}
\rput(65,65){\input{cube}}
\rput(65,75){\input{cube}}

\rput(75,35){\input{cube}}
\rput(75,45){\input{cube}}
\rput(75,55){\input{cube}}
\rput(75,65){\input{cube}}

\rput(5,45){\psline[linewidth=0.3pt,linecolor=black](0,-40)(0,50)}
\rput(15,45){\psline[linewidth=0.3pt,linecolor=black](0,-40)(0,50)}
\rput(25,45){\psline[linewidth=0.3pt,linecolor=black](0,-40)(0,50)}
\rput(35,45){\psline[linewidth=0.3pt,linecolor=black](0,-40)(0,50)}
\rput(45,45){\psline[linewidth=0.3pt,linecolor=black](0,-40)(0,50)}
\rput(55,45){\psline[linewidth=0.3pt,linecolor=black](0,-40)(0,50)}
\rput(65,45){\psline[linewidth=0.3pt,linecolor=black](0,-40)(0,50)}
\rput(75,45){\psline[linewidth=0.3pt,linecolor=black](0,-40)(0,50)}
\rput(85,45){\psline[linewidth=0.3pt,linecolor=black](0,-40)(0,50)}
\rput(95,45){\psline[linewidth=0.3pt,linecolor=black](0,-40)(0,50)}

\rput(45,5){\input{hline}}
\rput(45,15){\input{hline}}
\rput(45,25){\input{hline}}
\rput(45,35){\input{hline}}
\rput(45,45){\input{hline}}
\rput(45,55){\input{hline}}
\rput(45,65){\input{hline}}
\rput(45,75){\input{hline}}
\rput(45,85){\input{hline}}
\rput(45,95){\input{hline}}

\pscircle[linecolor=black,linewidth=1.1](50,50){30}
\rput(50,47){${\mathbf{s}}$}
\pscircle[linecolor=white](50,50){1}


\pscircle[linecolor=black,linestyle=dashed,dash=5pt 5pt](50,50){22.929}

\pscircle[linecolor=black,linestyle=dashed,dash=4pt 5pt](50,50){37.071}

\end{pspicture}
\end{center}
\caption{An illustration of Lemma~\ref{lem:gridsphere}. The solid circle is the boundary of $\mathcal{B}(\bs,r)$, and the points inside the small bright circles are the members of the set $\ZZ^n\cap\mathcal{B}(\bs,r)$. The set $\mathcal{}=\ZZ^n\cap\mathcal{B}(\bs,r)+\cube$ is the shaded area, and as the lemma indicates, it contains $\mathcal{B}(\bs,r-\frac{\sqrt{n}}{2})$ and is contained in $\mathcal{B}(\bs,r+\frac{\sqrt{n}}{2})$, whose boundaries are plotted in dashed circles.} \label{fig:spheregrid}
\end{figure}

%

\begin{lemma}
\label{lem:gridsphere}
For any $\bs\in\RR^n$ and $r>0$ ,the number of points of $\ZZ^n$ inside $\mathcal{B}(\bs,r)$ can be bounded as
\begin{align}
\left(\max\left\{r-\frac{\sqrt{n}}{2},0\right\}\right)^n V_n\leq \left|\ZZ^n\cap\mathcal{B}(\bs,r)\right|\leq \left(r+\frac{\sqrt{n}}{2}\right)^n V_n\nonumber
\end{align}
\end{lemma}

\begin{proof}
Let $\mathcal{S}\triangleq\left(\ZZ^n\cap\mathcal{B}(\bs,r)\right)+\cube$, and note that $\left|\ZZ^n\cap\mathcal{B}(\bs,r)\right|=\Vol(\mathcal{S})$.
We have
\begin{align}
\mathcal{B}\left(\bs,r-\frac{\sqrt{n}}{2}\right)\subseteq \mathcal{S}.\label{inclusion}
\end{align}
To see this, note that any $\bx\in\mathcal{B}\left(\bs,r-\frac{\sqrt{n}}{2}\right)$ lies inside \mbox{$\ba+\cube$} for some $\ba\in\ZZ^n$, and for this $\ba$ the inequality \mbox{$\|\ba-\bx\|\leq\sqrt{n}/2$} holds. Applying the triangle inequality gives
\begin{align}
\|\ba-\bs\|=\|(\ba-\bx)+(\bx-\bs)\|\leq\|(\ba-\bx)\|+\|(\bx-\bs)\|\leq r.\nonumber
\end{align}
Thus, $\ba\in\left(\ZZ^n\cap \m{B}(\bs,r)\right)$, and hence $\bx\in\m{S}$, which implies~\eqref{inclusion}.
On the other hand,
\begin{align}
\mathcal{S}&\subseteq \mathcal{B}(\bs,r)+\cube\nonumber\\
&\subseteq \mathcal{B}(\bs,r)+\mathcal{B}\left(0,\frac{\sqrt{n}}{2}\right)\nonumber\\
&=\mathcal{B}\left(\bs,r+\frac{\sqrt{n}}{2}\right).\nonumber
\end{align}
Thus,
\begin{align}
\Vol\left(\mathcal{B}\left(\bs,r-\frac{\sqrt{n}}{2}\right)\right)\leq \Vol\left(\mathcal{S}\right)\leq\Vol\left(\mathcal{B}\left(\bs,r+\frac{\sqrt{n}}{2}\right)\right).\nonumber
\end{align}
\end{proof}

\subsection{The matrix $\bG$ is full rank with high probability}
The probability that $\bG$ is not full-rank was bounded in~\cite{elz05}. We repeat the proof for completeness. The matrix $\bG$ is not full rank if and only if there exist some nonzero vector $\bw\in\ZZ_p^k$ such that $\bw^T\bG=\mathbf{0}$. Thus,
\begin{align}
\Pr\left(\rank(\bG_f)<k\right)&=\Pr\left(\bigcup_{\bw\in\ZZ_p^k\setminus\mathbf{0}} (\bw^T\bG=\mathbf{0})\right)\nonumber\\
&\leq \sum_{\bw\in\ZZ_p^k\setminus\mathbf{0}}\Pr(\bw^T\bG=\mathbf{0})\label{UBrank}\\
&=(p^k-1)p^{-n}\label{rankProb}\\
&<p^{-(n-k)},\nonumber
\end{align}
where~\eqref{UBrank} follows from the union bound, and~\eqref{rankProb} since $\bw^T\bG$ is uniformly distributed over $\ZZ_p^n$ for any $\bw\neq \mathbf{0}$.

By our definition of $p$ and $k$, and using the fact that $V_n^{\frac{2}{n}}\geq \tfrac{4}{n}$ for all $n$, we have
\begin{align}
k&\leq \frac{n}{2\log \xi+3\log{n}}\left( \log{n}+\alpha\right)\nonumber\\
&\leq n\left(\frac{1}{3}+\frac{\alpha}{3\log{n}}\right)\nonumber\\
&<\frac{2n}{3},\nonumber
\end{align}
where the last inequality follows from the assumption $\alpha<\log{n}$. Thus, $\Pr\left(\rank(\bG_f)<k\right)<p^{-\tfrac{n}{3}}$, and can therefore be made smaller than any $\epsilon>0$, by taking $n$ large enough.

\subsection{Goodness for MSE Quantization}
\label{subsec:MMSE}
In this subsection we show that for any $\delta,\epsilon>0$ and $n$ large enough
\begin{align}
\Pr\left(\sigma^2\left(\gamma\Lambda(\bG)\right)>(1+\delta)2^{-\alpha}\right)<\epsilon.\nonumber
\end{align}
Our proof follows the derivation from~\cite{yk07}, which dealt with the NSM of Construction A lattices with finite $p$. In our case $p$ grows with the lattice dimension, and the derivation can be significantly simplified.

We begin by bounding the average MSE distortion attained by the random lattice $\gamma\Lambda(\bG)$ for a source uniformly distributed over $\gamma[0,1)^n$. As we shall see, this average MSE distortion is equal to $\mathbb{E}(\sigma^2(\gamma\Lambda(\bG)))$. We then apply Markov's inequality to show that this implies that almost all lattices in the ensemble have a small $\sigma^2(\gamma\Lambda(\bG))$.

For any (fixed) $\bx\in\RR^n$, define
\begin{align}
d(\bx,\gamma\Lambda(\bG))&\triangleq\frac{1}{n}\min_{\lambda\in\gamma\Lambda(\bG)}\|\bx-\lambda\|^2\nonumber\\
&=\frac{1}{n}\min_{\ba\in\ZZ^n,\bc\in\mathcal{C}(\bG)}\|\bx-\gamma p^{-1}\bc-\gamma\ba\|^2\nonumber\\
&=\frac{1}{n}\min_{\bc\in\mathcal{C}(\bG)}\|(\bx-\gamma p^{-1}\bc)^*\|^2.\nonumber
\end{align}
Recall that $\gamma\ZZ^n\subset\gamma\Lambda(\bG)$ and therefore $d\left(\bx,\gamma\Lambda(\bG)\right)\leq \gamma^2/4$ for any $\bx\in\RR^n$, regardless of $\bG$.

Let $0<\rho<\alpha$. For any $\bw\in\ZZ_p^{k}\setminus{\mathbf{0}}$, define the random vector \mbox{$\bC(\bw)=\left[\bw^T\bG\right]\bmod p$}, and note that $\bC(\bw)$ is uniformly distributed over $\ZZ_p^n$. For all \mbox{$\bw\in\ZZ_p^{k}\setminus{\mathbf{0}}$} and $\bx\in\RR^n$, we have
\begin{align}
\varepsilon&\triangleq\Pr\left(\frac{1}{n}\left\|\left(\bx-\gamma p^{-1}\bC(\bw)\right)^*\right\|^2\leq 2^{-\rho}\right)\nonumber\\
&=p^{-n}\left|(\gamma p^{-1}\ZZ_p^n)\bigcap\mathcal{B}^*(\bx,\sqrt{n2^{-\rho}}) \right|\nonumber\\
&=p^{-n}\left|(\gamma p^{-1}\ZZ^n)\bigcap\mathcal{B}(\bx,\sqrt{n2^{-\rho}}) \right|\label{modisok}\\
&\geq p^{-n}V_n\left(p\gamma^{-1}\sqrt{n2^{-\rho}}-\frac{\sqrt{n}}{2}\right)^n\label{coveringPeBoundTmp}\\
&=V_n (\gamma^{-2}n2^{-\rho})^{\frac{n}{2}}\left(1-\frac{\gamma\sqrt{2^{\rho}}}{2p}\right)^n\nonumber\\
&=V_n \left(\frac{1}{4}\right)^{\frac{n}{2}}2^{-\tfrac{\rho n}{2}}\left(1-\frac{\sqrt{n2^{\rho}}}{p}\right)^n,\label{coveringPeBound}
\end{align}
where~\eqref{modisok} follows since $\gamma=2\sqrt{n}$, and hence, for any two distinct points $\bb_1,\bb_2\in\mathcal{B}(\bx,\sqrt{n2^{-\rho})})$ we have $\bb_1^*\neq\bb_2^*$ (that is, the ball $\mathcal{B}(\bx,\sqrt{n2^{-\rho})})$ is contained in a cube with side $\gamma$), and~\eqref{coveringPeBoundTmp} follows from Lemma~\ref{lem:gridsphere}. Substituting $p=\xi n^{\frac{3}{2}}$ and recalling that $\rho<\alpha<\log{n}$ gives
\begin{align}
\varepsilon&>2^{-\frac{n}{2}\left(\log\left(\frac{4}{V_n^{n/2}}\right)+\rho \right)}\left(1-\frac{2^{\tfrac{\rho}{2}}}{\xi n}\right)^n\nonumber\\
&>2^{-\frac{n}{2}\left(\log\left(\frac{4}{V_n^{n/2}}\right)+\rho \right)}\left(1-\frac{2^{\tfrac{\rho}{2}}}{n}\right)^n\nonumber\\
&>2^{-\frac{n}{2}\left(\log\left(\frac{4}{V_n^{n/2}}\right)+\rho \right)}\left(1-\frac{1}{\sqrt{n}}\right)^n\nonumber\\
&=2^{-\frac{n}{2}\left(\log\left(\frac{4}{V_n^{n/2}}\right)+\rho \right)}2^{n\log\left(1-\tfrac{1}{\sqrt{n}}\right)}\nonumber\\
&>2^{-\frac{n}{2}\left(\log\left(\frac{4}{V_n^{n/2}}\right)+\rho \right)}2^{-\frac{n\log(e)}{\sqrt{n}-1}},\nonumber
\end{align}
where we have used the inequality $\log(1-t)>-\tfrac{t}{1-t}\log(e)$ for $0<t<1$ in the last inequality. Thus, for any $n\geq 4$ we have
\begin{align}
\varepsilon>2^{-\frac{n}{2}\left(\log\left(\frac{4}{V_n^{n/2}}\right)+\rho +\frac{4\log(e)}{\sqrt{n}}\right)}.\label{coveringPeBound2}
\end{align}
Let $M\triangleq p^{k}-1$. Label each of the vectors $\bw\in\ZZ_p^{k}\setminus{\mathbf{0}}$ by an index \mbox{$i=1,\ldots,M$}, and refer to its corresponding codeword as $\bC_i$. Define the indicator random variable related to the point $\bx\in\RR^n$
\begin{align}
\chi_i=\begin{cases}
                   1 & \text{if }\frac{1}{n}\left|\left(\bx-\gamma p^{-1}\bC_i\right)^*\right|^2\leq 2^{-\rho} \\
                   0 & \text{otherwise}
                 \end{cases}.\nonumber
\end{align}
Since each $\chi_i$ occurs with probability $\varepsilon$, we have
\begin{align}
\Pr\bigg(\sum_{i=1}^M\chi_i=0\bigg)&=\Pr\left(\frac{1}{M}\sum_{i=1}^M\chi_i-\varepsilon=-\varepsilon\right)\nonumber\\
&\leq\Pr\left(\left|\frac{1}{M}\sum_{i=1}^M\chi_i-\varepsilon\right|\geq\varepsilon\right)\nonumber\\
&\leq\frac{\Var\left(\frac{1}{M}\sum_{i=1}^M\chi_i\right)}{\varepsilon^2},\label{chebyshev}
\end{align}
where the last inequality follows from Chebyshev's inequality. In order to further bound the variance term from~\eqref{chebyshev}, we note that $\bC(\bw_1)$ and $\bC(\bw_2)$ are statistically independent unless \mbox{$\bw_1=[a\bw_2]\bmod p$} for some $a\in\ZZ_p$. Therefore, each $\chi_i$ is statistically independent of all but $p$ different $\chi_j$'s. Thus,
\begin{align}
\Var\left(\frac{1}{M}\sum_{i=1}^M\chi_i\right)&=\frac{1}{M^2}\sum_{i=1}^M\sum_{j=1}^M\Cov(\chi_i,\chi_j)\nonumber\\
&\leq \frac{Mp\varepsilon}{M^2}.\nonumber
\end{align}
Substituting into~\eqref{chebyshev} and using~\eqref{coveringPeBound2}, we see that for any \mbox{$\bx\in\RR^n$}
\begin{align}
\Pr\bigg(d(\bx,&\gamma\Lambda(\bG))>2^{-\rho}\bigg)\leq \Pr\bigg(\sum_{i=1}^M\chi_i=0\bigg)\nonumber\\
&<\frac{p}{M\varepsilon}\nonumber\\
&<2n^{\frac{3}{2}}\frac{1}{p^{k}-1}2^{\frac{n}{2}\left(\log\left(\frac{4}{V_n^{n/2}}\right)+\rho +\frac{4\log(e)}{\sqrt{n}}\right)}\nonumber\\
&<4n^{\frac{3}{2}} p^{-k}2^{\frac{n}{2}\left(\log\left(\frac{4}{V_n^{n/2}}\right)+\rho +\frac{4\log(e)}{\sqrt{n}}\right)}\nonumber\\
&=4n^{\frac{3}{2}}2^{-\frac{n}{2}\left(\alpha-\rho -\frac{4\log(e)}{\sqrt{n}}\right)},\nonumber
\end{align}
where we have used
\begin{align}
p^{-k}=2^{-\frac{n}{2}\left(\log\left(\frac{4}{V_n^{n/2}}\right)+\alpha\right)}\nonumber
\end{align}
in the last equality.

It follows that for any distribution on $\bX$ we have
\begin{align}
&\mathbb{E}_{\bX,\bG}\left(d(\bX,\gamma\Lambda(\bG))\right)\nonumber\\
&\leq 2^{-\rho}\Pr\left(d(\bX,\gamma\Lambda(\bG))\leq 2^{-\rho}\right)\nonumber\\
&+\frac{\gamma^2}{4}\Pr\left(d(\bX,\gamma\Lambda(\bG))>2^{-\rho}\right)\nonumber\\
&\leq 2^{-\rho}\left(1+ 4n^{\frac{5}{2}}2^{-\frac{n}{2}\left(\alpha-\rho -\frac{4\log(e)}{\sqrt{n}}\right)}\right)\nonumber\\
&=2^{-\rho}\left(1+ 2^{-\frac{n}{2}\left(\alpha-\rho -\frac{4\log(e)}{\sqrt{n}}-\frac{5\log{n}}{n}-\frac{4}{n}\right)}\right).\nonumber
\end{align}
Thus, for any $0<\rho<\alpha$ the upper bound on the distortion averaged over $\bX$ and over the ensemble of lattices $\gamma\Lambda(\bG)$ becomes arbitrary close to $2^{-\rho}$ as $n$ increases.
Since this is true for all distributions on $\bX$, we may take $\bX\sim\Unif\left(\gamma[0,1)^n\right)$. Let $\bU$ be a random variable uniformly distributed over the Voronoi region $\CV_\bG$ of a lattice $\gamma\Lambda(\bG)$ randomly drawn from the ensemble. By construction, for any lattice $\gamma\Lambda(\bG)$ in the defined ensemble \mbox{$\left[\gamma p^{-1}\mathcal{C}(\bG)+\CV_\bG\right]^*=\gamma[0,1)^n$}. Moreover, reducing the set \mbox{$\gamma p^{-1}\mathcal{C}(\bG)+\CV_\bG$} modulo $\gamma\ZZ^n$ does not change its volume. Therefore,
\begin{align}
\mathbb{E}_{\bG}\left(\sigma^2(\gamma\Lambda(\bG))\right)=\mathbb{E}_{\bU,\bG}\left(\frac{1}{n}\|\bU\|^2\right)
=\mathbb{E}_{\bX,\bG}\left(d(\bX,\gamma\Lambda(\bG))\right).\nonumber
\end{align}
It follows that, for any $0<\rho<\alpha$,
\begin{align}
\mathbb{E}_{\bG}\left(\sigma^2(\gamma\Lambda(\bG))\right)\leq 2^{-\rho}\left(1+ 2^{-\frac{n}{2}\left(\alpha-\rho -\m{O}\left(\frac{1}{\sqrt{n}}\right)\right)}\right).\nonumber
\end{align}

Now, define the random variable $T\triangleq \sigma^2(\gamma\Lambda(\bG))-\tfrac{n}{n+2}2^{-\alpha}$. We show that the r.v. $T$ is non-negative, or equivalently, that for every $\bG$ in the ensemble
\begin{align}
\sigma^2(\gamma\Lambda(\bG))\geq\frac{n}{n+2}2^{-\alpha}.\label{momentLB}
\end{align}
To see this, note that $V(\gamma\Lambda(\bG))\geq\gamma^n p^{-k}$ for all $\bG$, with equality if and only if $\bG$ has full row rank. Thus, by~\eqref{normalizedVolume} we have $V^{\frac{2}{n}}(\gamma\Lambda(\bG))\geq n V_n^{\frac{2}{n}} 2^{-\alpha}$, which implies $r^2_{\text{eff}}(\gamma\Lambda(\bG))=V^{\frac{2}{n}}(\gamma\Lambda(\bG))/V_n^{\frac{2}{n}}\geq n 2^{-\alpha}$ by~\eqref{reffdef}. Using the isoperimetric inequality~\eqref{isoReff}, we get~\eqref{momentLB}.

Since $T$ is non-negative, we can apply Markov's inequality
\begin{align}
\Pr(T&>\delta 2^{-\alpha})\leq \frac{E(T)}{\delta}2^\alpha\nonumber\\
&=\frac{E(\sigma^2(\gamma\Lambda(\bG)))-\frac{n}{n+2}2^{-\alpha}}{\delta}2^\alpha\nonumber\\
&\leq\frac{2^{\alpha-\rho}\left(1+ 2^{-\frac{n}{2}\left(\alpha-\rho -\m{O}\left(\frac{1}{\sqrt{n}}\right)\right)}\right)-\frac{n}{n+2}}{\delta}\nonumber
\end{align}
Setting $\rho=\alpha-\log\left(\frac{n}{n+2}+\frac{\epsilon\delta}{2} \right)$ we get that for $n$ large enough $\Pr(T>\delta 2^{-\alpha})<\epsilon$, and therefore $\Pr\left(\sigma^2(\gamma\Lambda(\bG))>(1+\delta)2^{-\alpha}\right)<\epsilon$ as desired.

\subsection{Goodness for Coding}
\label{subsec:NN}

In this subsection we show that for any $\delta,\epsilon>0$, and any additive semi norm-ergodic noise $\bZ$ with effective variance $\sigma_\bZ^2=\frac{1}{n}\mathbb{E}\|\bZ\|^2\leq (1-\delta)2^{-\alpha}$, we have that
\begin{align}
\Pr\left(\Pr\left(Q_{\gamma\Lambda(\bG)}(\bt+\bZ)\neq \bt \mid \bG\right)>\delta \right)<\epsilon\nonumber
\end{align}
for any $\bt\in\gamma\Lambda(\bG)$, provided that $n$ is large enough.

For any $\bG$, we upper bound the error probability of the nearest neighbor decoder $Q_{\gamma\Lambda(\bG)}(\cdot)$ using the \emph{bounded distance} decoder, which is inferior. More precisely, we analyze the performance of a decoder that finds all lattice points of $\gamma\Lambda(\bG)$ within Euclidean distance $r$ from \mbox{$\bt+\bZ$}.
If there is a unique codeword in this set, this is the decoded codeword. Otherwise, the decoder declares an error. It is easy to see that regardless of the choice of $r$, the nearest neighbor decoder makes the correct decision whenever the bounded distance decoder does. Therefore, the error probability of the nearest neighbor decoder is upper bounded by that of the bounded distance decoder.

Given $\bG$, an error event $E$ for the bounded distance decoder can be expressed as the union of three events:
\begin{enumerate}
\item $E_1$ - The noise vector $\bZ$ falls outside a ball of radius $r$;
\item $E_2$ - The ball $\m{B}\left(\bt+\bZ,r\right)$ contains a point $\bt+\gamma\ba$ for some $\ba\in\ZZ^n\setminus\mathbf{0}$. This is equivalent to the event $\left(\m{B}\left(\bt+\bZ,r\right)\cap \left(\bt+\gamma\ZZ^n\right)\right)\setminus{\bt}\neq\emptyset$;
\item $E_3$ - The ball $\m{B}\left(\bt+\bZ,r\right)$ contains a point from $\gamma\Lambda(\bG)$ that does not belong to $\gamma\ZZ^n$. This is equivalent to the event $\m{B}\left(\bt+\bZ,r\right)\cap \left(\bt+\left(\gamma\Lambda(\bG)\setminus\gamma\ZZ^n\right)\right)\neq\emptyset$;
\end{enumerate}

Note that the first two events $E_1$ and $E_2$ depend only on $\bZ$, but not on $\bG$. Moreover,
\begin{align}
E_1=\left\{\bZ\notin\m{B}(0,r) \right\}\nonumber
\end{align}
and for $r<\gamma$ we can write
\begin{align}
E_2&=\left\{\left(\bZ+\m{B}(0,r)\right)\cap \left(\gamma\ZZ^n\setminus\mathbf{0}\right) \neq \emptyset \right\}\nonumber\\
&\subseteq \left\{\|\bZ\|+r\geq\gamma \right\}\nonumber\\
&=\left\{\bZ\notin\m{B}(0,\gamma-r) \right\}\nonumber.
\end{align}
In particular, if $\gamma>2r$ we have $E_2\subset E_1$. We choose $r^2=n\sqrt{1-\delta}2^{-\alpha}$ such that this condition indeed holds, and we can write
\begin{align}
\Pr(E\mid\bG)=\Pr(E_1\cup E_3\mid G)\leq \Pr(E_1)+\Pr(E_3\mid \bG).\label{PeG}
\end{align}
Thus,
\begin{align}
\Pr\left(\Pr(E\mid \bG)\geq \delta\right)&\leq \Pr\left(\Pr(E_1)+\Pr(E_3\mid \bG)\geq \delta\right)\nonumber\\
&=\Pr\left(\Pr(E_3\mid \bG)\geq \delta-\Pr(E_1)\right).\nonumber
\end{align}

Let $\delta'=\sqrt{\tfrac{1}{1-\delta}}-1>0$. We have for $\sigma^2_{\bZ}\leq(1-\delta)2^{-\alpha}$
\begin{align}
\Pr(E_1)&=\Pr\left(\bZ\notin\m{B}(0,r)\right)\nonumber\\
&=\Pr\left(\bZ\notin\m{B}\left(0,\sqrt{\frac{r^2}{n\sigma^2_\bZ}}\sqrt{n\sigma^2_{\bZ}}\right)\right)\nonumber\\
&\leq\Pr\left(\bZ\notin\m{B}\left(0,\sqrt{\frac{1}{\sqrt{1-\delta}}}\sqrt{n\sigma^2_{\bZ}}\right)\right)\nonumber\\
&=\Pr\left(\bZ\notin\m{B}\left(0,\sqrt{(1+\delta')n\sigma^2_{\bZ}}\right)\right)\nonumber.
\end{align}
Since $\bZ$ is semi norm-ergodic, it follows that $\Pr(E_1)<\delta/2$ for $n$ large enough.

Next, we turn to upper bounding $\Pr(E_3\mid\bG)$. Note that in contrast to $E_1$ and $E_2$, this event does depend on $\bG$. We therefore first show that $\mathbb{E}_\bG\left(\Pr(E_3|\bG)\right)$ is small, and then apply Markov's inequality to show that the probability of drawing a matrix $\bG$ for which $\Pr(E_3|\bG)>\delta/2$ is smaller than $\epsilon$.

Let $\Ind(\mathcal{A})$ be the indicator function of the event $\mathcal{A}$.
\begin{align}
\mathbb{E}_\bG&\left(\Pr(E_3|\bG)\right)=\mathbb{E}_{\bG}\left(\Pr\left(\left(\gamma\Lambda(\bG)\setminus\gamma\ZZ^n\right)\bigcap\mathcal{B}(\bZ,r)\neq\emptyset\right)\right)\nonumber\\
&=\mathbb{E}_{\bG}\mathbb{E}_{\bZ}\left(\Ind\left(\left(\gamma p^{-1}\mathcal{C}(\bG)\setminus{\mathbf{0}}\right)\bigcap\mathcal{B}^*(\bZ,r)\neq\emptyset\right) \ \big| \ \bG\right)\nonumber\\
&=\mathbb{E}_{\bZ}\mathbb{E}_{\bG}\left(\Ind\left(\left(\gamma p^{-1}\mathcal{C}(\bG)\setminus{\mathbf{0}}\right)\bigcap\mathcal{B}^*(\bZ,r)\neq\emptyset\right) \ \big| \ \bZ\right)\nonumber\\
&=\mathbb{E}_{\bZ}\Pr\left(\left(\gamma p^{-1}\mathcal{C}(\bG)\setminus{\mathbf{0}}\right)\bigcap\mathcal{B}^*(\bZ,r)\neq\emptyset \ \big| \ \bZ\right).
\end{align}
Since each codeword in $\mathcal{C}(\bG)\setminus{\mathbf{0}}$ is uniformly distributed over $\ZZ_p^n$, and there are less than $p^k$ such codewords (i.e., $p^k-1$), applying the union bound gives
\begin{align}
\mathbb{E}_\bG&\left(\Pr(E_3|\bG)\right)\leq \mathbb{E}_{\bZ}\left(p^{k-n}\cdot\left|\gamma p^{-1}\ZZ_p^n\bigcap\mathcal{B}^*(\bZ,r)\right| \ \bigg| \ \bZ \right)\nonumber\\
&\leq \mathbb{E}_{\bZ}\left(p^{k-n}\cdot\left|\gamma p^{-1}\ZZ^n\bigcap\mathcal{B}(\bZ,r)\right| \ \bigg| \ \bZ \right)\nonumber\\
&\leq p^{k-n}V_n\left(\frac{p}{\gamma}r+\frac{\sqrt{n}}{2}\right)^n\label{intersectionbound2}\\
&=p^k\gamma^{-n}V_n r^n\left(1+\frac{\gamma\sqrt{n}}{2p \ r}\right)^n\label{intersectionbound3}\\
&=\left(\frac{V_n^{\frac{2}{n}}}{\gamma^2 p^{-\frac{2k}{n}}}r^2\right)^{\frac{n}{2}}\left(1+\frac{1}{2p}\frac{\gamma 2^{\alpha/2}}{(1-\delta)^{1/4}}\right)^n\nonumber\\
&=\left(\frac{r^2}{n 2^{-\alpha}}\right)^{\frac{n}{2}}\left(1+\frac{1}{2p}\frac{\gamma 2^{\alpha/2}}{(1-\delta)^{1/4}}\right)^n\label{VolSubst}\\
&\leq (1-\delta)^{\frac{n}{4}}\left(1+\frac{ 2^{\alpha/2}(1-\delta)^{-1/4}}{n}\right)^n\nonumber\\
&\leq(1-\delta)^{\frac{n}{4}}e^{ 2^{\alpha/2}(1-\delta)^{-1/4}},\label{Vnbound}
\end{align}
where~\eqref{intersectionbound2} follows from Lemma~\ref{lem:gridsphere} and~\eqref{VolSubst} follows from~\eqref{normalizedVolume} and since $\gamma=2\sqrt{n}$. Now, by~\eqref{Vnbound} we have that $\mathbb{E}_\bG\left(\Pr(E_3|\bG)\right)<\delta\epsilon/2$ for $n$ large enough. Applying Markov's inequality gives that $\Pr\left(\Pr(E_3|\bG)>\delta/2 \right)<\epsilon$ as desired.

\section{Mixture Noise Is Semi Norm-Ergodic for MSE-Good Coarse Lattices}
\label{sec:mixture}

Our aim is to prove Theorem~\ref{thm:effball} that states that a mixture noise composed of semi norm-ergodic noise and a dither from a lattice that is good for MSE quantization, is semi norm-ergodic. First, we show that if the sequence $\Lambda^{(n)}$ is good for MSE quantization, i.e., its normalized second moment approaches $1/2\pi e$, then a sequence of random dithers uniformly distributed over $\CV^{(n)}$ is semi norm-ergodic. To that end, we first prove the following lemma, which is a simple extension of~\cite{ez02techreport}.

\vspace{1mm}

\begin{lemma}
\label{lem:effrad}
Let $\mathcal{S}\in\RR^n$ be a set of points with volume $V(\mathcal{S})$ and normalized second moment
\begin{align}
G(\mathcal{S})=\frac{1}{nV(\mathcal{S})}\frac{\int_{\mathcal{S}}\|\bx\|^2d\bx}{V(\mathcal{S})^{\frac{2}{n}}}.\nonumber
\end{align}
Let $r_{\text{eff}}$ be the radius of an $n$-dimensional ball with the same volume as $V(\mathcal{S})$, i.e., \mbox{$V(\mathcal{S})=V_n r_{\text{eff}}^n$}. For any $0<\epsilon<1$ define
\begin{align}
r_{\epsilon}\triangleq\sqrt{\frac{2\pi e G(\mathcal{S})-\frac{n}{n+2}(1-\epsilon)^{1+\frac{2}{n}}}{\epsilon}} r_{\text{eff}}.\nonumber
\end{align}
Then, the probability that a random variable \mbox{$\bU\sim\Unif(\mathcal{S})$} leaves a ball with radius $r_{\epsilon}$ is upper bounded by
\begin{align}
\Pr\left(\bU\notin \mathcal{B}(\mathbf{0},r_{\epsilon})\right)\leq\epsilon.\nonumber
\end{align}
\end{lemma}

\begin{proof}
Let $\tilde{r}_{\epsilon}$ be the radius of a ball that contains exactly a fraction of $1-\epsilon$ of the volume of $\mathcal{S}$, i.e.,
\begin{align}
\Vol\left(\mathcal{S}\bigcap\mathcal{B}(\mathbf{0},\tilde{r}_{\epsilon}) \right)=(1-\epsilon)V(\mathcal{S}).\nonumber
\end{align}
Clearly, $\Pr\left(\bU\notin \mathcal{B}(\mathbf{0},\tilde{r}_{\epsilon})\right)=\epsilon$. In order to establish the lemma we have to show that $\tilde{r}_{\epsilon}\leq r_{\epsilon}$.
To that end, we write
\begin{align}
n G(\mathcal{S}) V^{\frac{2}{n}}(\mathcal{S})&=\frac{1}{V(\mathcal{S})}\int_{\bx\in\mathcal{S}}\|\bx\|^2 d\bx \nonumber\\
&=\frac{1}{V(\mathcal{S})}\bigg(\int_{\bx\in\left(S\cap\mathcal{B}(\mathbf{0},\tilde{r}_{\epsilon})\right)}\|\bx\|^2 d\bx\nonumber\\
& \ \ \ \ \ \ \ \ \ \ + \int_{\bx\in\left(S\cap(\RR^n\setminus\mathcal{B}(\mathbf{0},\tilde{r}_{\epsilon}\right)}\|\bx\|^2 d\bx\bigg).\label{tmpIntegral}
\end{align}
The first integral in~\eqref{tmpIntegral} may be lower bounded by replacing its integration boundaries with an $n$-dimensional ball $\mathcal{B}(\mathbf{0},\rho_{\epsilon})$, where
\begin{align}
\rho_{\epsilon}^2=V_n^{-\frac{2}{n}}(1-\epsilon)^{\frac{2}{n}}V^{\frac{2}{n}}(\mathcal{S})
\end{align}
is chosen such that $V_n \rho_{\epsilon}^n=(1-\epsilon)V(\mathcal{S})$. Thus
\begin{align}
\int_{\bx\in\left(S\cap\mathcal{B}(\mathbf{0},\tilde{r}_{\epsilon}\right)}\|\bx\|^2 d\bx&\geq\int_{\bx\in\mathcal{B}(\mathbf{0},\rho_{\epsilon})}\|\bx\|^2 d\bx\nonumber\\
&=nV_n \rho_{\epsilon}^n\sigma^2\left({\mathcal{B}(\mathbf{0},\rho_{\epsilon})}\right)\nonumber\\
&=\frac{n}{n+2}V_n \rho_{\epsilon}^n \rho_{\epsilon}^2\label{secondmomentbound}\\
&=\frac{n}{n+2}\frac{V^{1+\frac{2}{n}}(\mathcal{S})(1-\epsilon)^{1+\frac{2}{n}}}{V_n^{\frac{2}{n}}}\nonumber\\
&=\frac{n}{n+2}V(\mathcal{S})(1-\epsilon)^{1+\frac{2}{n}}r_{\text{eff}}^2,\label{intbound1}
\end{align}
where we have used~\eqref{ballsecondmoment} to get~\eqref{secondmomentbound}.
The second integral in~\eqref{tmpIntegral} is over a set of points with volume $\epsilon V(\mathcal{S})$ which are all at distance greater than $\tilde{r}_{\epsilon}$ from the origin. Therefore, it can be bounded as
\begin{align}
\int_{\bx\in\left(S\cap(\RR^n\setminus\mathcal{B}(\mathbf{0},\tilde{r}_{\epsilon}))\right)}\|\bx\|^2d\bx\geq \epsilon V(\mathcal{S})\tilde{r}_{\epsilon}^2.\label{intbound2}
\end{align}
Substituting~\eqref{intbound1} and~\eqref{intbound2} into~\eqref{tmpIntegral} gives
\begin{align}
n G(\mathcal{S}) V^{\frac{2}{n}}(\mathcal{S})\geq\left(\frac{n}{n+2}(1-\epsilon)^{1+\frac{2}{n}}r_{\text{eff}}^2+ \epsilon \tilde{r}_{\epsilon}^2 \right).\label{repsilonbound}
\end{align}
Using the fact that $V^{\frac{2}{n}}(\mathcal{S})=V_n^{\frac{2}{n}}r_{\text{eff}}^2$,~\eqref{repsilonbound} reduces to
\begin{align}
\tilde{r}^2_{\epsilon}&\leq\frac{n V_n^{\frac{2}{n}}G(\mathcal{S})-\frac{n}{n+2}(1-\epsilon)^{1+\frac{2}{n}}}{\epsilon} r_{\text{eff}}^2\nonumber\\
&\leq\frac{2\pi e G(\mathcal{S})-\frac{n}{n+2}(1-\epsilon)^{1+\frac{2}{n}}}{\epsilon} r_{\text{eff}}^2\nonumber,
\end{align}
as desired.
\end{proof}

\vspace{1mm}

Using Lemma~\ref{lem:effrad} we can prove the following.

\vspace{1mm}

\begin{lemma}
\label{lem:dither}
Let $\Lambda^{(n)}$ be a sequence of lattices that is good for MSE quantization. Then the sequence of random dither vectors \mbox{$\bU^{(n)}\sim\Unif(\CV^{(n)})$} is semi norm-ergodic.
\end{lemma}

\vspace{1mm}

\begin{proof}
We need to show that for any $\epsilon,\delta>0$ and $n$ large enough
\begin{align}
\Pr\left(\bU^{(n)}\notin\m{B}(0,\sqrt{(1+\delta)n\sigma^2\left(\Lambda^{(n)}\right)}\right)\leq \epsilon.\nonumber
\end{align}
By Lemma~\ref{lem:effrad}, it suffices to show that
\begin{align}
&\sqrt{\frac{2\pi e G\left(\Lambda^{(n)}\right)-\frac{n}{n+2}(1-\epsilon)^{1+\frac{2}{n}}}{\epsilon}} r_{\text{eff}}\left(\Lambda{(n)}\right)\nonumber\\
& \ \ \ \ \ \ \ \ \ \ \ \leq \sqrt{(1+\delta)n\sigma^2\left(\Lambda^{(n)}\right)}.\label{rineq}
\end{align}
From~\eqref{reffbound}, we have
\begin{align}
r_{\text{eff}}\left(\Lambda^{(n)}\right)\leq\sqrt{(n+2)\sigma^2\left(\Lambda^{(n)}\right)}.\label{reffbound2}
\end{align}
and the LHS of~\eqref{rineq} can be therefore upper bounded by
\begin{align}
\sqrt{n\sigma^2\left(\Lambda^{(n)}\right)}\sqrt{\frac{\frac{n+2}{n}2\pi e G\left(\Lambda^{(n)}\right)-(1-\epsilon)^{1+\frac{2}{n}}}{\epsilon}}\label{sqrtineq1}
\end{align}
The sequence of lattices $\Lambda^{(n)}$ is good for MSE quantization, and therefore for any $\delta_1>0$ and $n$ large enough
\begin{align}
G(\Lambda^{(n)})<(1+\delta_1)\frac{1}{2\pi e}.\nonumber
\end{align}
Setting $\delta_1=\delta\epsilon/3$, we have that for $n$ large enough
\begin{align}
\frac{n+2}{n}&2\pi e G\left(\Lambda^{(n)}\right)-(1-\epsilon)^{1+\frac{2}{n}}\nonumber\nonumber\\
&\leq \frac{n+2}{n}\left(1+\frac{\delta\epsilon}{3}\right)-(1-\epsilon)^{1+\frac{2}{n}}\nonumber\\
&\leq \epsilon+\delta\epsilon,\label{sqrtineq2}
\end{align}
where the last inequality follows since for $n$ large enough $\tfrac{n+2}{n}(1+\tfrac{\delta\epsilon}{3})<1+\tfrac{2\delta\epsilon}{3}$ and $(1-\epsilon)^{1+\frac{2}{n}}>1-\epsilon-\tfrac{\delta\epsilon}{3}$.
Combining~\eqref{sqrtineq1} and~\eqref{sqrtineq2} establishes~\eqref{rineq}.
\end{proof}

\vspace{1mm}

We are now ready to prove Theorem~\ref{thm:effball}.

\vspace{1mm}

\begin{proof}[Proof of Theorem~\ref{thm:effball}]
Since $\bN$ and $\bU$ are statistically independent, the effective variance of $\bZ$ is
\begin{align}
\sigma^2_{\bZ}=\frac{1}{n}\mathbb{E}\|\bZ\|^2=\alpha^2\sigma_{\bN}^2+\beta^2\sigma^2_{\bU}.\nonumber
\end{align}
We have to prove that
for any $\epsilon>0$, $\delta>0$ and $n$ large enough
\begin{align}
\Pr&\left(\bZ\notin\mathcal{B}(\mathbf{0},\sqrt{(1+\delta)n\sigma^2_{\bZ}}\right)<\epsilon.\nonumber
\end{align}
For any $\epsilon>0$, $\delta>0$ and $n$ large enough we have
\begin{align}
&\Pr\left(\bZ\notin\mathcal{B}(\mathbf{0},\sqrt{(1+\delta)n\sigma^2_{\bZ}}) \right)\nonumber\\
&=\Pr\left(\|\bZ\|^2>(1+\delta)n\sigma^2_{\bZ}\right)\nonumber\\
&=\Pr\left(\|\bN\|^2>(1+\delta)n\sigma^2_{\bN}\right)\nonumber\\
& \ \ \ \cdot\Pr\left(\|\bZ\|^2>(1+\delta)n\sigma^2_{\bZ} \ \big| \ \|\bN\|^2>(1+\delta)n\sigma^2_{\bN} \right)\nonumber\\
&+\Pr\left(\|\bN\|^2\leq(1+\delta)n\sigma^2_{\bN}\right)\nonumber\\
& \ \ \ \cdot\Pr\left(\|\bZ\|^2>(1+\delta)n\sigma^2_{\bZ} \ \big| \ \|\bN\|^2\leq(1+\delta)n\sigma^2_{\bN} \right)\nonumber\\
&\leq \frac{\epsilon}{3}+\Pr\bigg(\beta^2\|\bU\|^2+2\alpha\beta\bN^T\bU\nonumber\\
& \ \ \ \ \ \ \ \ \ \ \ \ \ \ \ \ \ >(1+\delta)n\beta^2\sigma_{\bU}^2 \ \big| \ \|\bN\|^2\leq(1+\delta)n\sigma^2_{\bN} \bigg)\label{useprop}\\
&\leq \frac{\epsilon}{3}+\Pr\left(\beta^2\|\bU\|^2>n\beta^2\sigma_{\bU}^2(1+\delta/2) \right)\nonumber\\
&\ \ \ \ \ +\Pr\left(2\alpha\beta\bN^T\bU>n\beta^2\sigma_{\bU}^2\delta/2 \ \big| \ \|\bN\|^2\leq(1+\delta)n\sigma^2_{\bN} \right)\label{unionbound}\\
&\leq \frac{2\epsilon}{3}+\Pr\left(2\alpha\beta\bN^T\bU>n\beta^2\sigma_{\bU}^2\delta/2 \ \big| \ \|\bN\|^2\leq(1+\delta)n\sigma^2_{\bN} \right),\label{peboundtmp}
\end{align}
where~\eqref{useprop} follows from the fact that $\bN$ is semi norm-ergodic,~\eqref{unionbound} from the union bound and~\eqref{peboundtmp} from the fact that $\bU$ is semi norm-ergodic due to Lemma~\ref{lem:dither}. We are left with the task of showing that the last probability in~\eqref{peboundtmp} can be made smaller than $\epsilon/3$ for $n$ large enough. This requires some more work.

Since $\bU$ is semi norm-ergodic noise, than for any $\epsilon_2>0$, $\delta_2>0$ and $n$ large enough
\begin{align}
\Pr\left(\|\bU\|>\sqrt{(1+\delta_2)n\sigma_{\bU}^2} \right)<\epsilon_2.\nonumber
\end{align}
Let $r_{\bU}=\sqrt{(1+\delta_2)n\sigma_{\bU}^2}$, and $f_{\bU}(\bu)$ be the probability density function (pdf) of $\bU$. For any $r>0$ we have
\begin{align}
\Pr&\left(\bN^T\bU>r \ \big| \bN=\bn\right)=\int_{|\bu|\leq r_{\bu}}f_{\bU}(\bu)\Ind(\bn^T\bu>r)d\bu\nonumber\\
&\ \ \ \ \ \ \ \ \ \ \ \ \ \ \ \ \ \ \ \ \ \ \ \ \ \ +\int_{|\bu|> r_{\bu}}f_{\bU}(\bu)\Ind(\bn^T\bu>r)d\bu\nonumber\\
&\leq \int_{|\bu|\leq r_{\bu}}\frac{1}{V(\Lambda)}\Ind(\bn^T\bu>r)d\bu+\epsilon_2\nonumber\\
&=\frac{V(\mathcal{B}(\mathbf{0},r_{\bU}))}{V(\Lambda)}\int_{|\bu|\leq r_{\bu}}\frac{1}{V(\mathcal{B}(\mathbf{0},r_{\bU}))}\Ind(\bn^T\bu>r)d\bu+\epsilon_2.\nonumber
\end{align}
Using the fact that $\Lambda$ is good for MSE quantization we have $V(\Lambda)^{\frac{2}{n}}\rightarrow 2\pi e \sigma^2_{\bU}$, and hence, for $n$ large enough,
\begin{align}
\left(\frac{V(\mathcal{B}(\mathbf{0},r_{\bU}))}{V(\Lambda)}\right)^{\frac{2}{n}}<(1+2\delta_2).\nonumber
\end{align}
Let $\tilde{\bU}$ be a random vector uniformly distributed over $\mathcal{B}(\mathbf{0},r_{\bU})$. We have
\begin{align}
\Pr&\left(\bN^T\bU>r \ \big| \bN=\bn\right)
<\epsilon_2+(1+2\delta_2)^{\frac{n}{2}}\Pr(\bn^T\tilde{\bU}>r).\label{inprod}
\end{align}
Let $\tilde{\bZ}$ be AWGN with zero mean and variance $r^2_{\bU}/n$. Using a similar approach to that taken in~\cite[Lemma 11]{ErezZamirAWGN}, we would now like to upper bound the pdf of $\tilde{\bU}$ using that of $\tilde{\bZ}$. For any $\bx\in\RR^n$ we have
\begin{align}
\frac{f_{\tilde{\bU}}(\bx)}{f_{\tilde{\bZ}}(\bx)}=\frac{f_{\tilde{\bU}}(\|\bx\|)}{f_{\tilde{\bZ}}(\|\bx\|)}\leq\frac{f_{\tilde{\bU}}(r_{\bU})}{f_{\tilde{\bZ}}(r_{\bU})}=\left(\frac{2\pi e}{nV_n^{\frac{2}{n}}} \right)^{\frac{n}{2}}.\nonumber
\end{align}
Thus, for any $\bx\in\RR^n$
\begin{align}
f_{\tilde{\bU}}(\bx)\leq 2^{\frac{n}{2}\log\left(\frac{2\pi e}{n} V_n^{-\frac{2}{n}}\right)}f_{\tilde{\bZ}}(\bx)\nonumber.
\end{align}
We can further bound~\eqref{inprod} for large enough $n$ as
\begin{align}
\Pr&\left(\bN^T\bU>r \ \big| \bN=\bn\right)\nonumber\\
&\leq\epsilon_2+2^{\frac{n}{2}\log\left((1+2\delta_2)\frac{2\pi e}{n} V_n^{-\frac{2}{n}}\right)}\Pr(\bn^T\tilde{\bZ}>r)\nonumber\\
&=\epsilon_2+2^{\frac{n}{2}\log\left((1+2\delta_2)\frac{2\pi e}{n} V_n^{-\frac{2}{n}}\right)}Q\left(\frac{\sqrt{n}r}{\|\bn\|r_{\bU}}\right),\nonumber
\end{align}
where $Q(\cdot)$ is the standard $Q$-function, which satisfies \mbox{$Q(x)<e^{-x^2/2}$}. It follows that
\begin{align}
&\Pr\left(2\alpha\beta\bN^T\bU>n\beta^2\sigma_{\bU}^2\delta/2 \ \big| \ \|\bN\|^2\leq(1+\delta)n\sigma^2_{\bN} \right)\nonumber\\
&\leq \epsilon_2\nonumber\\
&+2^{\frac{n}{2}\log\left((1+2\delta_2)\frac{2\pi e}{n} V_n^{-\frac{2}{n}}\right)}Q\left(\frac{\sqrt{n}\beta\sigma_{\bU}\delta/2}{2\alpha\sigma_{\bN}\sqrt{(1+\delta)(1+2\delta_2)}}\right).\nonumber
\end{align}
Taking $\delta_2$ sufficiently smaller than $\delta$ and $\epsilon_2<\epsilon/6$, for $n$ large enough we have
\begin{align}
\Pr\left(2\alpha\beta\bN^T\bU>n\beta^2\sigma_{\bU}^2\delta/2 \ \big| \ \|\bN\|^2\leq(1+\delta)n\sigma^2_{\bN} \right)<\frac{\epsilon}{3}.\nonumber
\end{align}
\end{proof}

We end this section with two simple corollaries of Theorem~\ref{thm:effball}. The first follows since any i.i.d. noise is semi norm-ergodic, and the second follows by iterating over Theorem~\ref{thm:effball}.

\begin{corollary}
\label{cor:mixturenoise}
Let $\bZ=\alpha\bN+\beta\bU$, where $\alpha,\beta\in\RR$, $\bN$ is an i.i.d. noise vector, and $\bU$ is a dither statistically independent of $\bN$, uniformly distributed over the Voronoi region $\CV$ of a lattice $\Lambda$ that is good for MSE quantization. Then, the random vector $\bZ$ is semi norm-ergodic.\end{corollary}
\vspace{1mm}

\begin{corollary}
\label{cor:cofmixturenoise}
Let \mbox{$\bU_1,\cdots,\bU_K$} be statistically independent dither random vectors, each uniformly distributed over the Voronoi region $\CV_k$ of $\Lambda_k$, $k=1,\ldots,K$, that are all good for MSE quantization. Let $\bN$ be a semi norm-ergodic random vector statistically independent of \mbox{$\left\{\bU_1,\cdots,\bU_K\right\}$}. For any \mbox{$\alpha,\beta_1,\cdots,\beta_K\in\RR$} the random vector \emph{$\bZ=\alpha\bN+\sum_{k=1}^K\beta_k\bU_k$} is semi norm-ergodic.
\end{corollary}

\section{Nested Lattice Codes with a Cubic Coarse Lattice}
\label{sec:noshaping}

In this section we prove Theorem~\ref{thm:cubic}. As before, we consider an ensemble of $p$-ary random Construction A lattices. More precisely, we draw a matrix $\bG\in\ZZ_p^{k\times n}$ with i.i.d. entries uniformly distributed over $\ZZ_p$, and construct the (random) lattice $\gamma\Lambda(\bG)$ as in Definition~\ref{def:constA}, with $\gamma=\sqrt{12\Tsnr}$. We take $\gamma\Lambda(\bG)$ as a fine lattice and $\gamma\ZZ^n\subset\gamma\Lambda(\bG)$ as a coarse lattice, to construct the nested lattice codebook $\m{L}=\gamma\Lambda(\bG)\cap\gamma\cube$. Clearly, $\sigma^2\left(\gamma\ZZ^n\right)=\Tsnr$ and the rate of all codebooks in the ensemble is $R=\tfrac{k}{n}\log{p}$.

Applying the mod-$\Lambda$ scheme with the codebook $\m{L}$, as described in the proof of Theorem~\ref{thm:lapidoth}, gives rise to the effective channel~\eqref{Yeff}, where $\bZe$ is as defined in~\eqref{effNoise}. Note that for the coarse lattice $\gamma\ZZ^n$ which is used, the random vector $\bX$ is i.i.d. with each component uniformly distributed over $[-\gamma/2,\gamma/2)$. Thus, $\bZe$ is i.i.d., and in particular semi norm-ergodic, with variance $\sigma^2_{\bZe}(\alpha)=\alpha^2+(1-\alpha)^2\Tsnr$. As in the proof of Theorem~\ref{thm:lapidoth}, we choose $\alpha=\Tsnr/(1+\Tsnr)$ such as to minimize $\sigma^2_{\bZe}(\alpha)$, which gives $\sigma^2_{\bZe}=\Tsnr/(1+\Tsnr)$. As in~\eqref{codetNN}, the decoder finds
\begin{align}
\hat{\bt}=\left[Q_{\gamma\Lambda(\bG)}(\bYe)\right]\Mod_c=\left[Q_{\gamma\Lambda(\bG)}(\bt+\bZe)\right]\bmod \gamma\ZZ^n,\nonumber
\end{align}
and outputs the message corresponding to $\hat{\bt}$. In order to complete the proof we will need the following lemma.

\begin{lemma}
Let $n$ be a natural number, $p$ a prime number, $R>0$, $k=nR\log{p}$ and $\gamma>0$. Let $\bG\in\ZZ_p^{k\times n}$ be a random matrix with i.i.d. entries uniformly distributed over $\ZZ_p$, and $\Lambda(\bG)$ be constructed as in Definition~\ref{def:constA}. Let $\bZ$ be an additive semi norm-ergodic noise with effective variance $\sigma_\bZ^2=\frac{1}{n}\mathbb{E}\|\bZ\|^2$ and define $\Gamma(p,\gamma^2/\sigma^2_{\bZ})\triangleq\log\left(1+\sqrt{\frac{\gamma^2}{4p^2 \sigma^2_\bZ}} \right)$. For any $\epsilon,\delta>0$ and $n$ large enough, if $R<\tfrac{1}{2}\log\left(\frac{\gamma^2}{(1+\delta)2\pi e\sigma_{\bZ}^2}\right)-\Gamma(p,\gamma^2/\sigma^2_{\bZ})$, then
    \begin{align}
    \Pr\left(\Pr\left(Q_{\gamma\Lambda(\bG)}(\bt+\bZ)\neq \bt \bmod \gamma\ZZ^n\mid \bG\right)>\delta \right)<\epsilon.\label{cosenNNPe}
    \end{align}
for any $\bt\in\gamma\Lambda(\bG)$.
\label{lem:cosetNN}
\end{lemma}

Note that in~\eqref{cosenNNPe}, the error probability in \emph{coset nearest neighbor decoding} is required to be smaller than $\delta$. In other words, the decoder is only required to find the correct coset $\gamma\Lambda(\bG)/\gamma\ZZ^n$ to which $\bt$ belongs, and not the exact point $\bt$ that was transmitted. See Figure~\ref{fig:cosetdecoder} for an illustration of coset nearest neighbor decoding.

Theorem~\ref{thm:cubic} now follows by applying Lemma~\ref{lem:cosetNN} with $\gamma=\sqrt{12\Tsnr}$, $\sigma^2=\Tsnr/(1+\Tsnr)$ and taking $\delta$ to zero. This shows that for every $\delta>0$, for almost every $\bG$ and $n$ large enough, the error probability of the mod-$\Lambda$ scheme with codebook $\m{L}=\gamma\Lambda(\bG)\cap\sqrt{12\Tsnr}\cdot\cube$ is smaller than $\delta$. In particular, there exists a sequence of such codebooks with vanishing error probability.

It now only remains to prove Lemma~\ref{lem:cosetNN}.

\begin{proof}[Proof of Lemma~\ref{lem:cosetNN}]
The proof is similar to that of Theorem~\ref{thm:randomConstA}, part~\ref{goodForCoding}, with a few differences we now specify.

We upper bound the error probability of the coset nearest neighbor decoder with that of a bounded distance coset decoder. The latter finds all points of $\gamma\Lambda(\bG)$ in a ball of radius $r$ around the output $\bt+\bZe$ and outputs the list of all these points reduced modulo $\gamma\ZZ^n$. If the list of cosets does not contain exactly one point, an error is declared. It can be verified that an error event $E$ of this decoder is the union of $E_1$ and $E_3$, defined in Section~\ref{subsec:NN}. The event $E_2$ that was defined there, corresponds to decoding a point different than $\bt$ inside the same coset as $\bt$. This event does not incur an error for coset nearest neighbor decoding. Thus, equation~\eqref{PeG} continues to hold here.

We take the decoding radius as $r^2=n(1+\delta)\sigma^2_{\bZ}$, such that by the semi norm-ergodicity of $\bZe$, it follows that for $n$ large enough $\Pr(E_1)<\delta/2$. In order to upper bound $\Pr(\Pr(E_3\mid\bG))$ we upper bound $\mathbb{E}_{\bG}(\Pr(E_3\mid\bG))$ and then apply Markov's inequality. By~\eqref{intersectionbound3} we have
\begin{align}
\mathbb{E}_{\bG}&(\Pr(E_3\mid\bG))\leq p^k\gamma^{-n}V_n r^n\left(1+\frac{\gamma\sqrt{n}}{2p \ r}\right)^n\nonumber\\
&=2^{n\left(R+\frac{1}{2}\log\left(V_n^{\frac{2}{n}}\frac{r^2}{\gamma^2}\right)+\log\left(1+\sqrt{\frac{n}{4p^2 }\frac{\gamma^2}{r^2}} \right) \right)}\nonumber\\
&\leq 2^{n\left(R+\frac{1}{2}\log\left(\frac{2\pi e}{n}\frac{n(1+\delta)\sigma^2_\bZ}{\gamma^2}\right)+\log\left(1+\sqrt{\frac{n}{4p^2}\frac{\gamma^2}{ n(1+\delta)\sigma^2_\bZ}} \right)\right)}\nonumber\\
&= 2^{-n\left(\frac{1}{2}\log\left(\frac{\gamma^2}{2\pi e (1+\delta)\sigma^2_\bZ}\right)-\log\left(1+\sqrt{\frac{\gamma^2}{4p^2 \sigma^2_\bZ}} \right)-R\right)},\nonumber
\end{align}
Thus, for any $R<\frac{1}{2}\log\left(\frac{\gamma^2}{2\pi e (1+\delta)\sigma^2_\bZ}\right)-\Gamma(p,\gamma^2/\sigma^2)$, we have that $\mathbb{E}_\bG\left(\Pr(E_3|\bG)<\delta\epsilon/2\right)$, for $n$ large enough. Applying Markov's inequality gives that $\Pr\left(\Pr(E_3|\bG)>\delta/2 \right)<\epsilon$ as desired.
\end{proof}

\section*{Acknowledgment}
The authors thank Bobak Nazer, Yair Yona and Ram Zamir for discussions that helped prompt this work.
\bibliographystyle{IEEEtran}
\bibliography{OrBib2}

\end{document}